\newif\ifabstract
\abstracttrue
 \abstractfalse 
\newif\iffull
\ifabstract \fullfalse \else \fulltrue \fi

\documentclass[11pt]{article}
\usepackage{natbib}
\usepackage[textsize=small]{todonotes}

\usepackage{algorithm}
\usepackage[noend]{algpseudocode}
\usepackage{bm}
\usepackage{amsmath,amsthm,amssymb}
\usepackage{mathtools}
\usepackage{comment}
\usepackage[capitalize,noabbrev]{cleveref}

\newcommand{\cM}{\mathcal{M}}
\newcommand{\abm}{(\alpha,\beta)\textsc{-BalancedMatching}}
\newcommand{\abb}{(\alpha,\beta)\textsc{-Balanced}}

\newcommand{\OPT}{\mathrm{OPT}}
\newcommand{\LPOPT}{\mathrm{LPOPT}}

\newcommand{\pfm}{\textsc{ProportionallyFairMatching}}

\newcommand{\paf}{\textsc{ProbablyFair}}

\newcommand{\sd}[1]{{\color{violet}{\bf [~Sharmila:\ } \color{violet}{\em #1}\color{violet}{\bf~]}}}

\newcommand{\Ber}{\mathrm{Ber}}

\newcommand{\scr}{\mathcal}
\newcommand{\eps}{\varepsilon}
\newcommand{\mb}{\mathbb}

\theoremstyle{plain}
\newtheorem{theorem}{Theorem}[section]

\newtheorem{lemma}[theorem]{Lemma}
\newtheorem{corollary}[theorem]{Corollary}
\newtheorem{remark}[theorem]{Remark}

\theoremstyle{definition}

\newtheorem{claim}{Claim}[section]

\theoremstyle{remark}
\newtheorem{observation}{Observation}[section]
\newtheorem{example}{Example}[section]

\newcommand{\bE}{\mathbb{E}}
\newcommand{\cH}{\mathcal{H}}

\newcommand{\tbeta}{\Tilde{\beta}}

\newcommand{\nut}{N_{\bar{u}}(v_t)}

\newcommand{\tF}{\Tilde{F}}

\newcommand{\auv}{\psi_c(u,v)}
\newcommand{\auvt}{\psi_c(u,v_t)}
\newcommand{\auvk}{\psi_c(u,v_k)}
\newcommand{\awvt}{\psi_c(w,v_t)}
\newcommand{\awvk}{\psi_c(w,v_k)}

\newcommand{\pf}[1]{ProbablyFair}

\DeclareMathOperator{\var}{Var}

\advance \oddsidemargin by -0.8in
\newcommand{\myparskip}{3pt}
\parskip \myparskip

\title{Proportionally Fair Matching via Randomized Rounding\footnote{A preliminary version of this work appeared in Proceedings of AAAI 2025.}}
\author{
Sharmila Duppala\textsuperscript{1} \and 
Nathaniel Grammel\textsuperscript{1} \and 
Juan Luque\textsuperscript{1} \and 
Calum MacRury\textsuperscript{2} \and 
Aravind Srinivasan\textsuperscript{1}
}
\date{}
\begin{document}
\maketitle
\vspace{-1cm}
\begin{center}
\textsuperscript{1}University of Maryland, College Park, College Park, MD, USA \\
\textsuperscript{2}Columbia University, New York, NY, USA
\end{center}
\begin{abstract}
      Given an edge-colored graph, the goal of the \emph{proportional fair matching} problem is to find a \emph{maximum weight matching} 
  while ensuring proportional representation (with respect to the number of edges) of each color. The colors may correspond to demographic groups or other protected traits where we seek to ensure
  roughly equal representation from each group.
  It is known that, assuming ETH, it is impossible to approximate the problem with $\ell$ colors in time $2^{o(\ell)} n^{\mathcal{O}(1)}$ (i.e., subexponential in $\ell$) even on \emph{unweighted path graphs}. Further, even determining the existence of a non-empty matching satisfying proportionality is NP-Hard. 
  To overcome this hardness, we relax the stringent
  proportional fairness constraints to a probabilistic notion. We introduce a notion we call $\delta$-$\paf$, where we ensure proportionality up to a factor of at most $(1 \pm \delta)$ for some small $\delta >0$ with high probability. The violation $\delta$ can be brought arbitrarily close to $0$ for some \emph{good} instances with large values of matching size.  
  We propose and analyze simple and fast algorithms for bipartite graphs that achieve 
  constant-factor approximation guarantees, and return a $\delta$-$\paf$ matching. 
\end{abstract}

\section{Introduction}


Graph matching, in particular the special case of bipartite matching, is a classical computational problem with many applications such as ad allocation~\citep{mehta2007Adwords,mehta2013OnlineMatching}, crowdsourcing~\citep{ho2021Online,tong2016Online,hikima2021Matching,dickerson2019balancing}, job hiring~\citep{purohit2019Hiring}, organ exchange~\citep{dickerson2013failure,mcelfresh2019kidney,farnadi2021kidney,farhadi2022generalized}, and ride sharing~\citep{hikima2021Matching,nanda2020rideshare,dickerson2018ridesharing}.
Matching is also a fundamental subroutine in several domains including computer vision~\citep{belongie2002shape}, text similarity estimation~\citep{pang2016text} in natural language processing, machine learning algorithms~\citep{huang2007loopy,jebara2009graph,huang2011fast,choromanski2013adaptive} and computational biology~\citep{zaslavskiy2009global} among others.

In many applications, algorithms are employed to make decisions that could significantly impact the lives of individuals. In such settings, we ought to ensure fairness and equity in the decision-making process. However, classical algorithms typically do not consider such socially motivated objectives and constraints such as fairness or diversity. For instance, a ride-sharing platform may provide better quality assignments (e.g., with shorter wait times, newer vehicles, or better pricing) to riders from certain demographic groups while providing lower quality assignments to others as studied by \cite{esmaeili2023rawlsian}. 
Similarly, cognitive biases from workers can negatively impact the result of crowdsourced data, which may potentially be mitigated by assigning a diverse set of workers to a wide range of different \emph{task types}.

There are many ways to formalize and incorporate some notion of \emph{fairness} into the computed solutions of matching problems (see, e.g., \citet{esmaeili2023rawlsian,bandyapadhyay2023proportionally}). Following the work of~\citet{bandyapadhyay2023proportionally}, we consider a notion of \emph{proportional fairness}. This notion is closely related to the notion of \emph{disparate impact}~\citep{feldman2015disparate} and has been explored in various fundamental problems, including matroid optimization~\citep{chierichetti2019matroids}, clustering~\citep{bera2019fair}, online matching~\citep{sankar2021matchings}, set packing~\citep{duppala2023packing}, spectral clustering~\citep{kleindessner2019guarantees}, and others. In the present work, we consider a formulation of proportional fairness known as $\abm$: given an edge-colored graph (where edge colors may denote membership in specific groups), the objective is to find a maximum weight matching ensuring that the proportion of edges of any color among all matched edges lies between $\alpha$ and $\beta$ where $0\!\leq\! \alpha\! \leq\! \beta\! \leq \! 1$. 
We highlight that $\abm$ adheres to two key criteria: \emph{restricted dominance}, which limits the fraction of edges selected from any given group to at most $\beta$, and \emph{minority protection}, which ensures that the fraction of edges from any given group is at least $\alpha$.

This notion of fairness can be applied to the applications discussed above, such as in ride-sharing, job hiring, and crowd sourcing. By assigning colors to edges (rather than only to vertices), the \emph{color} is able to capture information about both sides of the match. For instance, in ridesharing, the edge color may encode information about both the driver (e.g., vehicle type, rating) and rider (e.g., race or gender), the distance between the driver and rider, and even pricing information (e.g., whether dynamic or surge pricing is applied). As a concrete example, we may choose to encode information about a rider's race or economic status as well as whether an assignment incurs increased fares due to dynamic pricing; then, an $\abm$ should ensure that no racial or economic group receives increased fares disproportionately often compared to other riders. In crowdsourcing, we may use edge color to encode both the task type and demographic information about the worker to ensure that a diverse set of workers are assigned to each type of task (for cognitive or human intelligence tasks, this may increase the diversity of perspectives among responses, and thus the overall quality of the final aggregate data).

As a final example application, consider online advertising: an edge's color may encode information about the type of advertisement (or, in political advertising, the political party sponsoring the ad) and the user's personal information. Here, an $\abm$ would limit the
extent of \emph{targeted} advertising based on protected demographic traits, which may violate (or
appear to violate) a user's expectations and rights with regards to their
privacy and use of their data; the social consequences of such indiscriminate highly targeted
advertising were seen in the real world with the Cambridge Analytica data
scandal around the 2016 US presidential election, which saw Facebook's CEO Mark
Zuckerberg testify before congress~\citep{guardianCambridgeAnalytica,nytCambridgeMZ}.

Unfortunately, although this notion of fairness seems quite powerful, it may in fact be too strong in the basic formulation: \citet{bandyapadhyay2023proportionally} demonstrated that $\abm$ is NP-hard to approximate, even when $G$ is an \emph{unweighted path graph}. 

To address this fundamental hardness, we consider a slightly relaxed probabilistic notion of proportional fairness. Roughly speaking, we allow for small violations of the $\alpha$ and $\beta$ bounds while still guaranteeing a constant factor approximation on the overall objective, i.e., the weight of the matching. Particularly, we ask the following question: \emph{does there exists an approximation algorithm with constant factor violations in fairness while guaranteeing a good approximation on the objective?}

We answer this question by designing a $\frac{1}{2}$-approximate algorithm for bipartite graphs
and whose matching is \emph{probably\footnote{i.e., with probability close to $1$} almost fair}. 
That is, there exists a small constant $ \delta >0$ for which the matching is $(\alpha(1-\delta),\beta (1+\delta))$-balanced with high probability. A formal definition of this notion is
given in~\cref{sec:prelim-prob}. 
Thus, while the $\abm$ formulation of~\citet{bandyapadhyay2023proportionally} remains hard on bipartite graphs, we can achieve our notion of \emph{probably almost fair}. {We note that for the special case when $\alpha =0$ and $\beta < 1$, we can attain a constant approximation ratio without violating the fairness constraint at all. Combined with the {hardness result} of \citet{bandyapadhyay2023proportionally}, this implies that the one-sided fairness case is {strictly easier} than the two-sided case.}

\section{Related Work}
A significant body of research in fair matching explored various notions of fairness. Among these, several notions of \emph{group fairness}, 
such as socially-fair matching~\citep{bandyapadhyay2023proportionally}, focus on Rawlsian (maxmin) fairness, aiming to maximize the utility for the worst-off group~\citep{esmaeili2023rawlsian}. 
On the other hand, proportional fairness, as explored in \citet{bandyapadhyay2023proportionally,chierichetti2019matroids}, ensures a proportional representation of edges from each group, 
while leximin fairness~\citep{garcia2020fair} is also considered. Most of these studies involve assigning group memberships to vertices in the graph. However, works such as \cite{bandyapadhyay2023proportionally,chierichetti2019matroids} investigate fairness in edge-colored graphs, which aligns with our area of study. {We note that  \cite{bandyapadhyay2023proportionally,chierichetti2019matroids} focus only on the unweighted setting where the objective is the cardinality/size of the output. That being said, our works are incomparable. While ours applies to edge weighted bipartite graphs, \cite{chierichetti2019matroids} consider when the constraint system is described by the intersection of two-matroids (this generalizes unweighted bipartite matching). They focus on the case of two colors,
and present a polynomial-time algorithm achieving a $3/2$-approximation. \citet{bandyapadhyay2023proportionally}
consider when input is a general graph and the number of colors may be greater than two.} They show that approximating the
problem for an arbitrary number of colors is NP-hard and moreover that
approximating the optimal solution requires time exponential in $\ell$ unless the Exponential Time Hypothesis (ETH)~\citep{impagliazzoETH} is false.
Their main algorithm therefore runs in time exponential in $\ell$, with an approximation guarantee of $1/(4\ell)$. Note that even for the easier case when $\alpha =0$, their algorithm may still violate the fairness constraint imposed by $\beta$ by a factor up to $1 + 1/\ell$.
Their hardness result motivates our study of the slightly relaxed probabilistic fairness. This relaxation allows us to improve on the prior work by achieving, in polynomial time, a \emph{constant factor} (i.e., one half) approximation in the matching weight while ensuring, with high probability, only a small violation in the fairness constraints. 

\section{Preliminaries and Problem Formulation}
\label{sec:prelim-prob}
We denote $[k] = \{1, \ldots, k\}$ for any positive integer $k$. 
Consider an undirected bipartite graph \( G = (U,V, E) \) with the set of edges \( E \) forming a partition \( \dot\cup_{c\in [\ell]}E_c \), where \( \dot\cup \) denotes a disjoint union over the sets \( E_c \). Each color \( c \in [\ell] \) is associated with the set of edges \( E_{c} \), representing edges of color \( c \). For any color class \( c \), we define \(\psi_c(e)\) such that \(\psi_c(e) = 1\) if \( e \in E_c \), and \(\psi_c(e) = 0\) otherwise.
A \textit{matching} \( M \subseteq E \) in \( G \) is defined as a subset of edges where no two edges in \( M \) share a common vertex. For a vertex \( v \) of $G$, let \( N(v) \) and $\delta(v)$ denote the neighbors and incident edges of \( v \), i.e., \( N(v) := \{ u : (u,v)\in E \} \) and \( \delta(v) := \{ (u,v) :(u,v)\in E \}\). 
For any $\alpha, \beta \in [0,1]$ and $\alpha\leq \beta$, we define a matching $M \subseteq E$ as $\abb$ if, for each color $c \in [\ell]$, the proportion of edges in $M$ belonging to color $c$ lies between $\alpha$ and $\beta$. In other words, $M$ is $\abb$ if it contains at least $\alpha$ and at most $\beta$ fraction of edges from every color.
The goal of the \emph{proportional fair matching problem} (PFM) is to find a $\abm$ of $G$ with maximum weight, and we denote the weight of such a matching
by $\OPT$. Since \cite{bandyapadhyay2023proportionally} proved that it is NP-Hard to verify the existence of a 
$\abm$ in the $\text{PFM}$ problem even on unweighted path graphs, we focus on deriving approximation ratios
which hold against $\OPT$.



\subsection{Our Contributions and Techniques}
We design \Cref{alg:ocrs}, an efficient randomized algorithm for weighted matching on bipartite
graphs with fairness constraints
specified by $0 < \alpha \le \beta <1$. By allowing
the $\alpha$ (respectively, $\beta$) fairness constraint to be violated up to a multiplicative factor of $1 - \delta$
(respectively, $1 + \delta$), we show that \Cref{alg:ocrs} 
attains an asymptotic approximation ratio against $\OPT$. (We write the exact asymptotic guarantee of \Cref{alg:ocrs} in \Cref{thm:thm_two_sided_fairness}). We argue that if the input is well-behaved, then
we can take $\delta \approx 0$, and so we barely violate the fairness constraints, while still attaining
a constant approximation ratio. We next analyze \Cref{alg:ocrs} when $\alpha\! =\!0$
and $\beta < 1$. By allowing for a $\frac{1-\eps}{2}$-approximation ratio, where $0 < \eps < 1$
is a parameter of \Cref{alg:ocrs} which can be chosen arbitrarily small, we show how to ensure that the $\beta$ fairness constraints are satisfied exactly. This allows us to get a true approximation ratio which is arbitrarily close to $1/2$. 
We additionally note that our algorithms can be extended to the slightly more general setting where we have different proportionality requirements for each color; e.g., where we have values $\alpha_c, \beta_c$ (with $0 \le \alpha \le \beta \le 1$) for each color $c\in[\ell]$, and we insist that $\alpha_c |\cM| \le |\cM_c| \le \beta_c |\cM|$ for each $c$. This allows us to capture settings where we wish to represent different groups in different proportions, for example so that each group's representation in the matching is roughly proportional to its representation in the original graph.



We take a linear programming (LP) based approach to designing \Cref{alg:ocrs}. Below we state \ref{lp:fairmat}, which extends the standard matching polytope to include the proportionality constraints described by $0 <  \alpha \le \beta <1$:
\begin{subequations} 
\begin{alignat}{3} 
      \max \sum_{ e \in E}& w_{e} x_{e} \tag{\textsc{Lp-Fair}} \label{lp:fairmat} \\[3ex]
    \text{s.t.\ } \sum_{e \in \delta(v)} &x_{e} \leq 1,  \,  \forall v \in  U \cup V \label{eq:left} \\
 \alpha  \sum_{ e \in E} x_{e}    \leq  \sum_{ e \in E_c} &x_{e} \leq \beta  \sum_{ e \in E} x_{e},  & & \, \forall c \in [\ell] \label{eq:coloring_constraint} \\
    &x_{e} \geq 0,   \,  \forall e \in E \label{eq:nonnegativity}
\end{alignat}
\end{subequations}
\begin{lemma} \label{lem:relaxation}
\ref{lp:fairmat} relaxes $\OPT$. That is, $\OPT \le \sum_{e \in E} w_e x_e$,
where $\bm{x} = (x_e)_{e \in E}$ is an optimal solution to \ref{lp:fairmat}.
\end{lemma}
As a first attempt to using \ref{lp:fairmat}, observe that
since $G$ is bipartite, if $\bm{x}=(x_e)_{e \in E}$ is an optimal solution to the LP,
then we can write $\bm{x}$ as a convex combination of integral matchings.
By randomly sampling such a matching according to the coefficients of
the convex combination, this yields a randomized algorithm with expected value
$\sum_{e \in E} w_e x_e \ge \OPT$. Unfortunately, while this preserves the fairness
constraints in expectation due to \eqref{eq:coloring_constraint}, there is no guarantee that any \textit{particular} matching we output will satisfy these constraints. In order for an algorithm
to succeed with constant probability, an ideal approach would be to match
each edge $e$ with probability $x_e$ while ensuring that the size of a matching of a color class $c$ is concentrated about $\sum_{e \in E_c} x_e$. For instance, if we had \textit{negative correlation} amongst the matching statuses of the edges of $E_c$, then this would be sufficient (see \citep{Dubhashi_Ranjan_1996}). The GKPS dependent rounding scheme of \citep{gandhi2006dependent} provides such a guarantee for \textit{certain} edge subsets of $G$. However, since the edges of $E_c$ are adversarially chosen, it is easy to construct an example where the GKPS rounding scheme induces positive correlation amongst the matched statuses.







Due to the limitations of these approaches,
we need to round our LP in a different way. Let us suppose that we round $\bm{x}$ into a random matching $\scr{M}$,
where we denote $\scr{M}_c := \scr{M} \cap E_c$ for a fixed color class $c \in [\ell]$. Roughly speaking, our goal is to ensure
that the following properties simultaneously hold, where $\delta > 0$ is a fixed parameter to be specified in \Cref{thm:one_sided_fairness}:
\begin{enumerate}
    \item For each $e \in E$, $\Pr[e \in \scr{M}] = x_{e}/2 $. \label{eqn:rounding_guarantee}
    \item With constant probability, 
    \[
        \frac{|\scr{M}_c|}{|\scr{M}|} \in \left[(1- \delta) \frac{\bE[|\scr{M}_c|]}{\mb{E}[|\scr{M}|]}, (1 + \delta)  \frac{\mb{E}[|\scr{M}_c|]}{\mb{E}[|\scr{M}|]}\right].
    \]  \label{eqn:concentration}
\end{enumerate}
Property \ref{eqn:rounding_guarantee} and 
\Cref{lem:relaxation} then imply that $$\mb{E}[\sum_{e \in \scr{M}} w_e] =  \sum_{e \in E} w_e x_e/2  \ge  \OPT/2.$$ Moreover, since $\mb{E}[|\scr{M}_c|] = \sum_{e \in E_c} x_e/2 $
and $\mb{E}[|\scr{M}|] = \sum_{e \in E} x_e/2$, we can combine Property \ref{eqn:concentration}. with \eqref{eq:coloring_constraint}
to conclude that $(1 - \delta) \alpha |\scr{M}| \le |\scr{M}_c| \le (1 +\delta) \beta |\scr{M}|$.

Our approach uses a randomized rounding tool known as
a \textit{contention resolution scheme (CRS)}. CRS's 
were originally introduced by \cite{chekuri2014submodular} to solve certain constrained sub-modular optimization problems. Motivated by the applications
to prophet inequalities, they were later adapted 
to the online setting by \cite{feldman2021online}, where they are referred to as \textit{online} contention resolution schemes (OCRS's). Since then,
they have found many other applications,
and have become a fundamental tool in stochastic optimization.
We continue this line of research and show that this tool can be useful
even for our problem which is offline and has no inherit stochasticity. 

Our algorithm first has every $v \in V$ independently draw a random vertex $F_v \in N(v)$, where 
\begin{equation} \label{eqn:proposal_draw}
\Pr[F_v = u] = x_{u,v}
\end{equation}
for each $u \in N(v)$. (For convenience, we define $F_{v}$ to be a \textit{null element} $\bot$ if no draw is made, where $\mb{P}[F_v = \bot] =1 - \sum_{u \in N(v)} x_{u,v}$). If $F_v = u$, we say that $v$ \textit{proposes}
to $u$ and refer to $F_v$ as a \textit{proposal} for $u$. At this point, we know that each vertex $v \in V$ has made at most one proposal. However, a fixed vertex $u \in U$ may have received multiple proposals,
and so we must resolve which proposal each vertex $u$ should take.
This is precisely the purpose of the OCRS, and we use the explicit scheme of \cite{ezra2022prophet}.
For a fixed $u \in U$, the input required of the OCRS is the edge values $(x_{u,v})_{v \in N(u)}$, together
with the proposals $(F_v)_{v \in N(u)}$. Moreover, in the contention resolution framework, $(F_v)_{v \in N(u)}$ must be independent. 
Under these conditions, the OCRS outputs at most one edge $(u,v)$ with $F_v =u$, 
while guaranteeing that for \textit{all} $v \in N(u)$
\begin{equation} \label{eqn:ocrs_guarantee}
\Pr[\text{$(u,v)$ is output} \mid F_v = u] = 1/2. 
\end{equation}
Due to guarantee of \eqref{eqn:ocrs_guarantee}, the OCRS is said to be $1/2$-\textit{selectable} in the literature.
By concurrently
running a separate execution of this OCRS for each $u \in U$, the matching we output satisfies Property \ref{eqn:rounding_guarantee}. 


We still need the explain why our algorithm satisfies Property \ref{eqn:concentration}. It turns out
our algorithm has a very simple description. First, we order the vertices of $V$ arbitrarily,
say $v_1, \ldots ,v_n$. Then, when $v_t$
is processed and proposes to $u$ (i.e., $F_{v_t} = u$), we draw an \textit{independent} random bit $A_{u,v_t}$
We match $(u,v_t)$ provided $u$ was not previously matched and $A_{u,v_t} =1$. We note that in regards to verifying Property \ref{eqn:concentration}, the specific Bernouli parameter of $A_{u,v_t}$ is not important. We simply require that $(A_e)_{e \in E}$ are drawn independently. 

If we let $\scr{M}$ be the output of the algorithm, and $\scr{M}_c = E_c \cap \scr{M}$,
then we can analyze the Doob martingale of $|\scr{M}_c|$. Here the Doob martingale is defined by revealing the partial matching after the vertices $v_1, \ldots ,v_t$ are processed. Due to the simplicity of
our algorithm, we can bound the one-step changes in the Doob martingale in a very precise way
that is related to $\sum_{e \in E_c} x_e$. Note that the standard Azuma–Hoeffding martingale concentration inequality does not allow for good bounds due to constant worst-case one-step changes (as we explain in \cref{sec:azuma_fails} and an example in \cref{ex:azuma_fails}) leading to $\Theta(n)$ over the entire process from $t\!=\!1$ to $t=n$. We show that by controlling the \textit{variance} of this martingale,
we can get stronger bounds via Freedman's inequality. 
More precisely, we can tighten the concentration on the random variable $\cM_c$ by using one-step variance which can be much smaller than the worse-case one-step changes since the events involving the bad cases occur with small probability. The total variance can be precisely computed and is upper bounded by $10\sum_{e\in E_c}x_e$ which is much smaller than $\Theta(n)$. The concentration helps us achieve \emph{almost} fairness, i.e., $|\cM_c|$ violates the fairness constraints on each side by at most a $\delta$-multiplicative factor for some small $\delta>0.$ 
Finally, when we only have \(\beta\)-sided fairness constraints i.e., \(\alpha=0, \beta <1\), we apply an LP \emph{perturbation trick} to recover a matching that satisfies the constraints \emph{exactly}, meaning that without any violations with only a $0<\epsilon<1$ multiplicative loss in the objective. However, it is important to note that our algorithm still has a probabilistic guarantee: with a certain probability, the matching returned will \textit{not} satisfy the fairness constraints. At the expense of additional runtime, we can improve the success probability to be arbitrarily close to $1$ by repeatedly running our algorithm
and taking the best solution which is feasible.

\section{Our Algorithm} \label{sec:alg-descr}
We begin by formally describing our algorithm to find a matching for a bipartite graph
$G=(U,V,E)$ with color classes $\cup_{c \in [\ell]} E_c$ and $\alpha \le \beta \le 1$. 
Note that our algorithm takes in a parameter $0 < \eps < 1$ which is only needed for the special
case when $\alpha =0$. 

For $\delta>0$, a random matching $M \subseteq E$ is said to be $\delta$-$\paf$ with respect to $0\!\leq\! \alpha \!\leq\! \beta\! \leq\! 1$ if for any color class $c$,    \footnote{Taking $\delta=1$ satisfies the \cref{eq:paf} trivially}
\begin{equation} \label{eq:paf} 
   \Pr\left[  (1-\delta) { \alpha }   \leq  \frac{|M_c|}{ |M| } \leq  (1+\delta)  \beta \right]  \geq 1- f_c(\delta,G).
\end{equation}
where $f_c(\delta,G)$ is a term that depends on the tunable parameter $\delta$ and the input instance. This definition allows us to violate the constraints by a small $\delta$-multiplicative factor which can be adjusted to attain a reasonable probability of success.









\begin{algorithm}[H]
\caption{OCRS Rounding Algorithm}
\label{alg:ocrs}
\begin{algorithmic}[1] 
\Require $G=(U,V,E)$ with color classes $(E_c)_{c \in [\ell]}$, $0 \le \alpha \le \beta \le 1$,
and $0< \eps < 1$.
\Ensure subset of edges forming a matching $\scr{M}$.

\State $\scr{M} \leftarrow \emptyset$

\State If $\alpha > 0$, compute an optimal solution $\bm{x} = (x_e)_{e \in E}$ to \ref{lp:fairmat}
with $0 < \alpha \le \beta \le 1$. Otherwise, compute an optimal solution $\bm{x}$ to \ref{lp:fairmat}
with $\beta$ in \eqref{eq:coloring_constraint} replaced by $\tilde{\beta} = (1 - \eps) \beta$.

\State Order the vertices of $V$ as $v_1, \ldots ,v_n$ arbitrarily.
\State Draw $(F_{v_t})_{t=1}^n$ as described in \eqref{eqn:proposal_draw}.

\For{$t=1,\ldots ,n$ with $F_{v_t} \neq \bot$}
\State Set $u:= F_{v_t}$, and $a_{u,v_t}:= \frac{1/2}{1 - (1/2) \sum_{i < t}x_{u,v_i}}$.
\State Draw $A_{u,v_t} \sim \Ber(a_{u,v_t})$ independently.
\State If $A_{u,v_t} =1$ and $u$ is currently unmatched,
add $(u,v_t)$ to $\scr{M}$
\EndFor
\State \Return $\scr{M}$.
\end{algorithmic}
\end{algorithm}

Our main results are as follows:
\begin{theorem} 
\label{thm:thm_two_sided_fairness}
\cref{alg:ocrs} returns a matching \(\cM\) that has an expected weight of at least \(\frac{1}{2}\OPT\) and is \(\delta\)-\(\paf\) with \( f_c(\delta, G) = 4\exp\left( \frac{-\delta^2\sum_{e\in E_c}x_e}{28}\right)\). 
\end{theorem}
\begin{remark}
The violations dependent on \(\delta\) in the fairness constraints are inversely exponential to the failure probability \(f_c(\delta, G)\). This implies that achieving arbitrarily small violations \(\delta \approx 0\) with a high success probability requires stronger assumptions on the term \(\sum_{e\in E_c}x_e\) i.e., for input instances with $\sum_{e\in E_c}x_e= \omega(1)$ we can ensure high success probability with \(\delta \approx 0\).
\end{remark}
We next state our improved guarantee for the special case when $\alpha = 0$. Recall
that in this setting, we can ensure that we do \textit{not} violate the  $\beta$-fairness constraint.

\begin{theorem} \label{thm:one_sided_fairness}
Given any $0 < \eps < 1$,
\Cref{alg:ocrs}  returns a matching that satisfies the $\beta$-fairness constraints with probability at least $1-f(\epsilon,G)$ where $f(\epsilon,G)= 2\exp \left(-\epsilon^2 \beta \frac{\sum_{e\in E}x_e}{28}\right)$, and has an expected weight of at least $ \frac{1}{2}(1-\epsilon)\OPT.$  
\end{theorem}
In Theorem~\ref{thm:one_sided_fairness}, the values \((x_e)_{e \in E}\) represent the optimal fractional solution to \ref{lp:fairmat} with \(\tbeta = \beta(1 - \epsilon)\). Additionally, \(f(\epsilon, G)\) denotes the probability that the algorithm fails to produce a matching that satisfies the \(\beta\)-fairness constraints \emph{exactly} i.e., with no violations. 

\begin{remark}
The probability that our algorithm fails decreases exponentially as a function of \(\epsilon\) and \(\beta \sum_{e\in E}x_e\). Therefore, to attain a failure probability $f(\epsilon,G)\approx 0$ and the approximation factor close to $1/2$, we require that \(\beta \sum_{e\in E}x_e \) to be sufficiently large, say \(\beta \sum_{e\in E}x_e \geq C \) for some constant $C$, say $C\geq 100$.
\end{remark}

\begin{remark}
For the instances where \( \beta \sum_{e \in E} x_e \leq C \), we address these cases separately using a brute force algorithm as described in \cref{sec:brute_force_alg}. 
This involves enumerating all feasible matchings that satisfy the fairness with respect to $\beta$ and having size at most $\frac{U^2}{L^2}\frac{C}{\beta(1-\epsilon)}$ where $U = \max_{e \in E}w_e$ and $L= \min_{e\in E}w_e.$ 
\end{remark}
For all non-degenerate instances, \(\frac{U}{L}\) is bounded by a constant. Hence, our algorithm runs in polynomial time if the term \(\frac{U^2}{L^2} \frac{C}{(1 - \epsilon) \beta}\) remains constant, given that \(\frac{U}{L}\)is bounded. However, in the worst-case, where \(\beta = \frac{1}{\ell}\), the running time will be exponential in \(\ell\). 



\section{Algorithm Analysis} \label{sec:alg-analysis}
Suppose that we are given an instance of proportional fair matching. We let $\bm{x} = (x_e)_{e\in E}$ denote an optimal solution to $\ref{lp:fairmat}$ used by \Cref{alg:ocrs}. Note that if $\alpha =0$,
then $\bm{x}$ is computed by replacing $\beta$ in \eqref{eq:coloring_constraint} with $\tilde{\beta} = (1 - \eps) \beta$.


Recall that the vertices of $V$ are processed in order $v_1, \ldots ,v_n$. Let us say that $u \in U$ is \textit{safe}
at time $t$, provided $u$ is \textit{not} matched to any of $v_1, \ldots ,v_{t-1}$.
Observe that $(u,v_t)$ is matched if and only if $u$ is safe at time $t$, $F_{v_t}= u$ and $A_{u,v_t}\!=\!1$. Let $Z_{u,v_t}$ be an indicator random variable for this event. 
For convenience, we define \(\tilde{F}_{v_t} := u\) provided \(F_{v_t}=u\) and \(A_{u, v_t}=1\).  Otherwise, \(\tilde{F}_{v_t} = \bot\). Observe then that
\begin{align} \label{eq:zuvt}
    \Pr[Z_{u,v_t}=1] &= \Pr[\tF_{v_t}=u, u\text{ is safe at }t ]
\end{align}
Let $Z(v_t)$ be an indicator random variable that an edge of $E_c$ is matched to vertex $v_t$ where $Z(v_t) = \sum \limits_{ \mathclap{{u \in N(v_t)}} }  \auvt Z_{u,v_t}$.  Recall that $\cM$ is the matching returned by our algorithm, therefore,
\begin{align*}
  |\cM_c| = \sum_{t\geq 1} \sum_{(u,v_t) \in \delta(v_t)}  \auvt Z_{u,v_t}
\end{align*}
Now, the instantiations of \((\tilde{F}_{v_i})_{i=1}^t\) are sufficient to determine which edges (if any) are matched to $v_1, \ldots ,v_t$. Thus, we shall work with $\scr{H}_t$, the sigma-algebra generated by \((\tilde{F}_{v_i})_{i=1}^t\), which intuitively corresponds to the history of the algorithm's execution after $v_1, \ldots ,v_t$ are processed. We define $\scr{H}_0$ to be the trivial sigma-algebra, where we do not condition on anything.

Fix a color class \( c \), and let \( M_t := \mathbb{E}[ |\mathcal{M}_c| \mid \mathcal{H}_t ] \) denote the expected number of edges from color class \( c \) in our matching, conditional on $\scr{H}_t$. 
Observe that \( M_0 = \mathbb{E}[|\cM_c|] \), and \( M_n\) is the actual number of edges chosen from the color class \( c \) at the end of our algorithm. Since we expose strictly more information in each round, i.e., 
formally, $\cH_{t-1}\subseteq \cH_{t}$ for each $t \ge 1$, the sequence $(M_t)_{t=0}^{n}$ is a martingale with respect to $(\scr{H}_t)_{t =0}^{n}$. (This is often referred to as a Doob martingale).

We aim to understand the size of the matching produced by our algorithm as each vertex \( v_t \) is processed.
More specifically, when $v_t$ is processed, our algorithm adds the edge $(u,v_t)$ to the current
matching if $u$ is safe and $\tF_{v_t}=u$. Observe that $M_t$ is a function of the current matching between $U$ and $\{v_1,v_2,\dots,v_{t}\}$. Moreover, for any $1 \leq t \leq n$, 
\begin{equation} \label{eq:def_mt}
    M_t = \bE \Big[ \sum_{k\geq 1 } Z(v_k) \mid \cH_{t} \Big].
\end{equation}
Our goal is to control the one step changes of our martingale, which will imply
that $M_n = |\cM_c|$ is concentrated about $M_0 = \mb{E}[|\scr{M}_c|]$. 



\subsection{Warmup: Azuma-Hoeffding inequality} 
\label{sec:azuma_fails}
To apply concentration bounds like the Azuma–Hoeffding inequality, the martingale must be well-behaved i.e., at any step $1 \le t \le n$, the martingale cannot change dramatically. Thus, we need to determine the quantities $c_t$ such that if \(\Delta M_t := M_t - M_{t-1}\) is the one-step change at step $t$, then $|\Delta M_t| \le c_t$. 

\begin{lemma}  \label{lem:one_step_worst_case}
$|\Delta M_t| = |M_t - M_{t-1}| \le {2}$ for all $1 \leq t \leq n$.
\end{lemma}

\begin{proof}
By \cref{lemma:one_step_one_vertex_diff} we know that the one-step change $|\Delta M_t|$ is given by 
\begin{align*}
    | M_t - M_{t-1} | &=   \max_{ q \in N(v_t)\cup \{\bot\}} |\bE[M_n \mid \cH_{t-1}, \tF_{v_t}= q] - \bE[M_n \mid \cH_{t-1} ] | \\
    = \begin{split}
        \max \Big\{  \max_{u \in N(v_t)}  |\bE[M_n \mid \cH_{t-1}, \tF_{v_t}= u] - \bE[M_n \mid \cH_{t-1} ]|, \\ |\bE[M_n \mid \cH_{t-1}, \tF_{v_t}= \bot] - \bE[M_n \mid \cH_{t-1} ]| \Big\}
    \end{split}  \\
    &\leq 2
    \end{align*}
The last inequality follows from \cref{lemma:wose_case_mt} and \cref{lemma:wose_case_mt_2}. 
\end{proof}

Let consider the one-step change in \(\Delta M_t\) when \(\tF_{v_t}\) is revealed. Since we can match at most one edge at each time step, we may miss the opportunity to match at most \emph{two} edges of a given color. Therefore, in the worst case, we have that \(\sum_{t} |\Delta M_t| = \Theta(n)\).  For example, consider the star graph, where a central vertex is connected to all other vertices. In each time step, if we reveal $\tF_{v_t}$,
then the change in the maximum matching will be at most $1$, as we can only match or unmatch one edge (we defer the full details of the example to \cref{ex:azuma_fails}).
To remedy this situation, we use a ``second order'' inequality for analyzing a martingale called Freedman's inequality: 
\begin{theorem}[\cite{freedman1975tail}] \label{thm:freedman}
Suppose $(M_t)_{t\in [n]}$ is a martingale such that $|\Delta M_t| \leq \Lambda$ for any $1 \leq t \leq n$, and $ \sum_{t\in [n]} \var[ \Delta M_t
\mid \cH_{t-1}] \leq \nu.$ Then, for all $\lambda > 0$, 
\begin{equation}
    \Pr \big[ |M_n - M_0| \leq  \lambda \big] \leq 2\exp \Big( \frac{-\lambda^2}{2(\nu+ \lambda \Lambda) } \Big)
\end{equation}
\end{theorem}
\begin{remark}
The variance term $\var(\Delta M_t \mid \cH_{t-1})$ reduces to $\var[ \Delta M_t \mid \cH_{t-1}]= \mb{E}[|\Delta M_t|^2 \mid \cH_{t-1}]$ since $\var[ \Delta M_t \mid \cH_{t-1}] := \mb{E}[|\Delta M_t|^2 \mid \cH_{t-1} ] - \mb{E}[\Delta M_t \mid \scr{H}_{t-1}]^2$ and $\mb{E}[\Delta M_t \mid \scr{H}_{t-1}] =0$ due to the martingale property.
\end{remark}
Freedman's inequality still depends on worst case one-step change $\Lambda$ i.e., an upper bound on $|\Delta M_t|$. However, its dependence on this parameter is less severe compared to the Azuma–Hoeffding inequality. The key difference is that Freedman's inequality depends on  $\mb{E}[ |\Delta M_t|^2 \mid \scr{H}_{t-1}]$, and so
we get to \textit{average} over the randomness in $\tF_{v_t}$. This is much smaller than than the pessimistic bound of \Cref{lem:one_step_worst_case}, as even conditional on $\scr{H}_{t-1}$, the probability that a worst-case change in  $|\Delta M_t|$ occurs is small.
\subsection{Concentration of $|\cM_c|$ via Freedman's inequality}
Freedman's inequality allows us to use the expected one-step changes instead of the worst case changes as shown in the \Cref{obs:var_ub} below. 
\begin{observation} \label{obs:var_ub}
    $\var(\Delta M_t \mid \cH_{t-1}) \leq 2 \bE[|\Delta M_t| \mid \cH_{t-1}]   $
\end{observation}
\begin{proof}
The variance of the one-step changes \(\Delta M_t := M_t - M_{t-1}\), given \(\cH_{t-1}\), simplifies to \(\var(\Delta M_t \mid \cH_{t-1}) = \mathbb{E}\left[ |\Delta M_t|^2 \mid \cH_{t-1} \right]\), because \(\mathbb{E}[\Delta M_t \mid \cH_{t-1}] = 0\) due to the martingale property. Recall that \( M_n = |\cM_c| \). By the definition of \( M_t \) and given that \(\cH_t = (\mathcal{F}_{v_i})_{i=q}^{t}\), we have:
 \begin{align*}
     \var(\Delta M_t \mid \cH_{t-1}) &=  \sum_{u \in N(v_t)\cup \{\bot\}} \Pr[ \tF_{v_t}=q] \big|  \bE \big[ M_n \mid \cH_{t-1},\tF_{v_t}= q \big]  -\bE \big[ M_n \mid \cH_{t-1} \big]  \big|^2  \\
     &\leq \sum_{u \in N(v_t)\cup \{\bot\}} 2 \Pr[ \tF_{v_t}=q] \big|  \bE \big[ M_n \mid \cH_{t-1},\tF_{v_t}= q \big]  -\bE \big[ M_n \mid \cH_{t-1} \big]  \big|  \\
     &=2 \bE \big[|\Delta M_t| \mid \cH_{t-1} \big]
 \end{align*} 

The inequality above holds because the worst case change in the expected matching belonging to color class $c$ is at most $2$ as given by \cref{lem:one_step_worst_case}.    
\end{proof}

We will prove two important facts (\cref{lemma:safety_k_u} and \cref{lemma:safety_k_u_2}) about the one-step changes that will be used in the rest of the analysis.

\begin{lemma}\label{lemma:safety_k_u}
For any step \( k > t \), \(w\in N(v_t)\) and  \( w \neq u \), if \( \tF_{v_t} = u \), i.e., \( v_t \) proposes to \( u \) and \( A_{u, v_t} = 1 \), then
\begin{equation*} 
    \Pr[w \text{ is safe at }k \mid \cH_{t-1}, \tF_{v_t}=u ] - \Pr[ w \text{ is safe at }k \mid \cH_{t-1} ] \leq x_{w,v_t}. 
\end{equation*}
\end{lemma}
\begin{proof}
   Any vertex $w \in U$ is safe at step $k$ if $\tF_{v_r} \neq w$ for any $1\leq r \leq k$. The probability that $w \in N(v_t)$ is safe conditional on $\cH_{t-1}$ can be either $0$ or $1$. Therefore, we can conclude that 
   \begin{align*}
       \Pr \big[ w \text{ is safe at }k
 \mid & \cH_{t-1},\tF_{v_t}= u \big]  -   \Pr \big[ w \text{ is safe at }k
 \mid \cH_{t-1} \big]\\
 &\leq \Pr \big[ \cap_{t\leq r\leq k} \tF_{v_r} \neq w  \mid \cH_{t-1}, \tF_{v_t}=u   \big]  - \Pr \big[ \cap_{t\leq r\leq k} \tF_{v_r} \neq w \mid \cH_{t-1} \big]  \\
 &=\prod_{t< r\leq k} \Pr \big[  \tF_{v_r} \neq w \mid \cH_{t-1}, \tF_{v_t}=u  \big]  - \prod_{t\leq r\leq k}\big[  \tF_{v_r} \neq w  \mid \cH_{t-1} \big]\\
  &\leq \prod_{t< r\leq k} \Pr \big[  \tF_{v_r} \neq w  \mid \cH_{t-1}   \big] (1-1 + x_{w,v_t}) \\
  &\leq x_{w,v_t}
   \end{align*} 
The second last inequality is due to the fact for any $r \geq t$ that $\Pr[\tF_{v_r} \neq w \mid \cH_{t-1}] \geq (1-x_{w,v_r}).$  
\end{proof}

\begin{lemma}\label{lemma:safety_k_u_2}
For any step \( k > t \), and \(w\in N(v_t)\) , if \( \tF_{v_t} = \bot \) i.e., $v_t$ is not matched to any vertex incident on $v_t$ then
\begin{equation*} 
    \Pr[w \text{ is safe at }k \mid \cH_{t-1}, \tF_{v_t}=\bot ] - \Pr[ w \text{ is safe at }k \mid \cH_{t-1} ] \leq x_{w,v_t}
\end{equation*}
\end{lemma}
\begin{proof}
For each vertex \( w \in N(v_t) \), if \(\tF_{v_t} = \bot\), a proof similar to that in \cref{lemma:safety_k_u} can be used. This is because, in both cases where \(\tF_{v_t} = u\) or \(\tF_{v_t} = \bot\), we have:
\[\Pr[\tF_{v_t} \neq w \mid \cH_{t-1}, \tF_{v_t} = u] = \Pr[\tF_{v_t} \neq w \mid \cH_{t-1}, \tF_{v_t} = \bot].\]
\end{proof}

We now prove \Cref{lemma:one_step_one_vertex_diff} which establishes a bound on the one-step changes in the random variable $ Z(v_k)$ at any step $t$ where $k \geq t.$ 

\begin{lemma} 
\label{lemma:one_step_one_vertex_diff}

For any step $1 \le t \le n$ and $u \in N(v_t)$,
\begin{itemize}
    \item \( \Big| \bE \big[Z(v_t)\! \mid\! \cH_{t-1}, \tF_{v_t}\!=\!u \big] \!-\! \bE\big[Z(v_t)\! \mid \!\cH_{t-1}\big] \Big| \leq \auvt  + \sum\limits_{\mathclap{w \in N(v_t)}} \awvt x_{w,v_t}  \)  \label{eq:one_step_zt}
    \item \( \left| \bE \big[Z(v_k) \mid \cH_{t-1}, \tF_{v_t}=u \big] - \bE\big[Z(v_k) \mid \cH_{t-1}\big]  \right|\leq  \auvk x_{u,v_k}+ \sum \limits_{\mathclap{w \in \nut}} \awvk x_{w,v_k} x_{w,v_t} \) if $k>t$. \label{eq:one_step_zk}
\end{itemize}
\end{lemma}

\begin{proof}
We begin by proving the first part of the lemma. 
Recall that $Z_{w,v_t}$ is the indicator random variable that the edge $(w,v_t)$ is selected in the matching $\cM.$ From the definitions of $Z(v_t)$ and $\tF_{v_t}$, the probability  $ \Pr[\tF_{v_t}=w] = \Pr[A_{w,v_t}=1, F_{v_t}=w] $. This implies that $ \Pr[\tF_{v_t}=w \mid \cH_{t-1}] \leq x_{w,v_t} $. Therefore we have
\begin{align*}
&\bE \big[ Z(v_t) \mid \cH_{t-1},\tF_{v_t}=u \big] - \bE \big[ Z(v_t) \mid \cH_{t-1} \big] \\ 
&=   \sum_{ w \in N(v_t)}  \awvt  \Big( \bE \big[ Z_{w,v_t} \mid \cH_{t-1},\tF_{v_t}=u \big] - \bE \big[ Z_{w,v_t} \mid \cH_{t-1} \big] \Big) 
\end{align*}
We apply the bounds established in \cref{lemma:wose_case_mt}, specifically \cref{eq:uvtu} and \cref{eq:wvtu}, and substitute them into the equation above. That implies 
\begin{align*}
\bE \big[ Z(v_t) \mid \cH_{t-1},\tF_{v_t}=u \big] - \bE \big[ Z(v_t) \mid \cH_{t-1} \big]
&\leq  \auvt- \auvt\frac{ x_{u,v_t}}{2} -  \sum_{w \in N_{\tilde{u}} (v_t)} \awvt  x_{w,v_t} \\
\end{align*}
Therefore, by applying triangle inequality we get the desired bound for the first part of the lemma as follows:
\[\left| \big[ Z(v_t) \mid \cH_{t-1},\tF_{v_t}=u \big] - \bE \big[ Z(v_t) \mid \cH_{t-1} \big] \right|  \leq  \sum_{ w \in N(v_t)}  \awvt { x_{w,v_t}}  +\auvt.\]

For the second part, where \( k > t \),
we consider the bounds from \cref{lemma:wose_case_mt} in \cref{eq:wvku} and \cref{eq:safety_k_u} to obtain the following, 
\begin{equation*}
 \left| \bE[Z(v_k) \mid \cH_{t-1}, \tF_{v_t}=u] - \bE[Z(v_k) \mid \cH_{t-1}] \right|   \leq \auvk x_{u,v_k}+ \sum_{ w \in \nut  }    \awvk  x_{w,v_k} x_{w,v_t} 
\end{equation*}
as desired. This concludes the proof of the lemma. 
\end{proof}  

Following this lemma, we can say that the one-step change in the size of the matching restricted to color class \( c \) at any step \( t \), conditioned on the event \( \tF_{v_t} =u \), is given by:
\begin{align*} 
   \bE[M_n \mid \cH_{t-1},\tF_{v_t}=u] - \bE[M_n \mid \cH_{t-1}]
   = \sum_{k \geq t}   \bE[ Z(v_k) \mid \cH_{t-1},\tF_{v_t}=u] - \bE[Z(v_k) \mid \cH_{t-1}]. 
\end{align*}
Therefore, by using the \Cref{lemma:one_step_one_vertex_diff} we can derive the following result about the expected one-step change in the size of the matching restricted to color class $c$ when  $\tF_{v_t}=u$. Notice that these one step differences are in the worst case; therefore they can be as large as $2$ as discussed in the previous section. However, when we compute the one-step variances, we get to take an expectation over the randomness of $\tF_{v_t}$ as shown in \cref{lemma:guts} and \cref{lemma:bt}. This allows us to bypass this worst-case bound of $2$.

\begin{lemma} \label{lemma:guts}
For any step $1 \le t \le n$ and $u \in N(v_t)$,
    \[ \sum_{t\in [n]} \sum_{ u \in N(v_t)}x_{u,v_t} |\Delta M_t(u)|  \leq 4M_0\]
    where $\Delta M_t(u)\!=\!\bE\big[M_n \mid \cH_{t-1},\tF_{v_t}=u\big] \!- \!\bE\big[M_n\mid \cH_{t-1}\big].$
\end{lemma}

\begin{proof}
Using \Cref{lemma:gut} and the fact \(\sum_{t\in[n]} \sum_{ u \in N(v_t)}\auvt x_{u,v_t} = \sum_{(u,v)\in E}\auv x_{u,v}  = M_0,\) we can say that
    \begin{align}
    \begin{split}
       &  \sum_{t\in [n]} \sum_{ u \in N(v_t)}x_{u,v_t}  
        \big| \bE[M_n \mid \cH_{t-1},\tF_{v_t}=u] - \bE[M_n \mid \cH_{t-1}] \big| \\
         &\leq M_0 +  \sum_{t\ge 1 }\Big( \sum_{ u \in N(v_t)} x_{u,v_t}   \sum \limits_{w \in \nut}  \big(  \awvt x_{w,v_t}   +    \sum_{k >t}  \awvk x_{w,v_k}x_{w,v_t}  \big)  + \notag \\
          &\quad \sum_{u \in N(v_t)}x_{u,v_t}\auvk x_{u,v_k}  \Big)
         \end{split} \notag \\
          &\leq  2M_0 +\sum_{(u,v)\in E} \sum_{w \in N(v)} \psi_c({w,v}) x_{u,v} x_{w,v}  + \sum_{t\ge 1}\sum_{u \in N(v_t)} x_{u,v_t}    \sum \limits_{w \in \nut} \sum_{k >t}  \awvk x_{w,v_k}x_{w,v_t} \label{eq:one} \\
        &\leq 2M_0 + \sum_{(w,v)\in E_c}x_{w,v} \sum_{u  \in N(v)} x_{u,v}    +  \sum_{ (u,v) \in E}x_{u,v} \sum_{w \in N_{\bar{u}}(v)} \sum_{v' \in N(w)} \psi_c({w,v'}) x_{w,v'} x_{w,v} \label{eq:event_e} \\
         &\leq 2M_0 +   \sum_{(w,v)\in E_c} x_{w,v}  + \sum_{ (u,v) \in E} \sum_{w \in N(v)}  \sum_{{ (w,v') \in E_c\cap \delta(w) }}  x_{w,v}x_{w,v'} x_{u,v} \label{eq:all_edges} \\
          &\leq  2M_0 + M_0+ \sum_{ (w,v') \in E_c}x_{w,v'} \sum_{v \in N(w)}x_{w,v}  \sum_{u\in N(v)} x_{u,v} \label{eq:critical_inequality} \\
         &= 3M_0 +\sum_{ (w,v') \in E_c}x_{w,v'}  \Big(\sum_{v \in N(w)}x_{w,v} \big(\sum_{u \in N(v) } x_{u,v}\big) \Big) \notag  \\
         &\leq 3M_0 + \sum_{ (w,v') \in E_c}x_{w,v'} \leq 4 M_0 \notag
    \end{align}
\cref{eq:one} is due the fact that
\begin{align*}
    \sum_{t\geq1}\sum_{u \in N(v_t)} x_{u,v_t}  \auvk x_{u,v_k} & \leq \sum_{(u,v) \in E}x_{u,v} \sum_{v' \in N(u)} \psi_c(u,v')x_{u,v'} \\ & \leq \sum_{(u,v') \in E_c} x_{u,v'} \sum_{v \in \delta(u)}x_{u,v} \leq \sum_{(u,v') \in E_c} x_{u,v'}=M_0
\end{align*}

 Note that each $\tF_{v_t} =u$ indicates that the edge \((u, v_t)\) satisfies $F_{v_t}=u$ and  $A_{u,v_t}=1$. Therefore, the summation \(\sum_{t \in [n]} \sum_{u \in N(v_t)}\) simplifies to a summation over the set of edges \( E \) as shown in \cref{eq:event_e}. We obtain \cref{eq:all_edges} due to the fact that \(\sum_{s > t} \sum_{w \in \nut} a_{w,v_s}\) can be upper bounded by summing over all edges of color class \( c \).  The inequality~\eqref{eq:critical_inequality} is true because we change the order of summations from $\sum_{(u,v)\in E}\sum_{w \in N(v)}\sum_{(w,v') \in E_c \cap \delta(w)}$ without changing the value of the summation. The final inequality follows from the fact that \(\left( \sum_{v \in N(w)} x_{w,v} \left( \sum_{u \in N(v)} x_{u,v} \right) \right) \leq 1\) for any \( w \in U \). This concludes the proof.
\end{proof}

We next consider the case when $\tF_{v_t} = \bot$. Recall that this means $F_{v_t} = \bot$, or $F_{v_t} = u$ for some $u \in N(v_t)$, yet $A_{u,v_t} = 0$.
\begin{lemma} \label{lemma:bt}
For any step $1 \le t \le n$,
\[\sum_{t\in [n]} | \bE \big[M_n \mid \cH_{t-1},\tF_{v_t} = \bot\big] \!- \!\bE\big[ M_n \mid \cH_{t-1}\big]| \leq 2M_0.\]
\end{lemma}

\begin{proof}
In \cref{lemma:wose_case_mt_2} we showed that the worst case one-step change  $\Delta M_t$ when $\tF_{v_t} = \bot$ occurs. Therefore, we directly apply \cref{eq:botmn} to get
$\sum_{ u \in N(v_t) \setminus \{u\}} \sum_{k \geq t}  \auvk  x_{u,v_k} x_{u,v_t} +  \sum_{u \in N(v_t)}\auvt x_{u,v_t}$
 
\begin{align}
      &\sum_{t\geq 1} \Big| \bE[M_n \mid \cH_{t-1},\tF_{v_t} = \bot] - \bE[ M_n \mid \cH_{t-1}]  \Big| \notag  \\
     &\leq   \sum_{t\geq 1} \sum_{u \in \delta(v_t)}\auvt x_{u,v_t} + \sum_{t\ge 1}\sum_{k> t} \sum_{ u \in N(v_t)} \auvk x_{u,v_k} x_{u,v_t}  \notag  \\
        &= M_0 + \sum_{t \in [n]} \sum_{u \in N(v_t)}x_{u,v_t}  \sum_{(u,v_k)\in E_c:k>t}x_{u,v_k} \label{eq:final_f}
\end{align}
Finally, we can say that,
    \begin{align}
       \sum_{t \in [n]} \sum_{u \in N(v_t)}x_{u,v_t}  \sum_{(u,v_k)\in E_c:k>t}x_{u,v_k}  
       &\leq   \sum_{t \in [n]} \sum_{u \in N(v_t)}x_{u,v_t}  \sum_{v':(u,v')\in E_c}x_{u,v'} \notag \\
         &\leq   \sum_{(u,v)\in E}x_{u,v} \sum_{ (u,v') \in E_c \cap \delta(u) }x_{u,v'} \notag  \\
         & = \sum_{(u,v') \in E_c} x_{u,v'} \sum_{v \in N(u)}x_{u,v} \notag \\
         &\leq \sum_{(u,v') \in E_c} x_{u,v'} =M_0. \label{eq:final_e}
    \end{align}
    Therefore, we get the desired bound by combining  \cref{eq:final_f} and \cref{eq:final_e}. 
\end{proof}

Therefore, by combining \cref{lemma:guts} and \cref{lemma:bt}, we can easily derive a bound on the sum of one-step variances i.e., \(\sum_{t \in [n]} \var(\Delta M_t \mid \cH_{t-1})\) used in the Freedman's inequality. Formally, we can prove the following.
\begin{lemma} \label{eq:var_t_sum}
$ \sum_{t\in [n]}\var(\Delta M_t \mid \cH_{t-1}) \leq 12M_0 $ where $\var(\Delta M_t \mid \cH_{t-1})$ is the variance in the one-step change of $M_n$ at time $t$ and $M_0=\bE[|\cM_c|]$. 
\end{lemma}
\begin{proof}
    Letting $x(t) = \sum_{e \in \delta(v_t)}x_e$, we have the following. 
    \begin{align}
        \bE[ |\Delta M_t| \mid \cH_{t-1} ]
        &= \sum_{u \in \tF_{v_t}\cup \{\bot \} } \Pr[\tF_{v_t} = u] \Big| \bE[M_n \mid \cH_{t-1},\tF_{v_t}=u ] - \bE[M_n \mid \cH_{t-1}]  \Big| \label{eq:indep_1} 
    \end{align} 
    Inequality~\eqref{eq:indep_1} is by definition and the fact that the random variable $\tF_{v_t}$ is independent of $\cH_{t-1}.$ The support of $\tF_{v_t}$ comprises of mainly two types of events (i) an edge $(u,v_t) \in \delta(v_t)$ has been both selected ($F_{v_t}=u$) and survived $A_{u,v_t}=1$, and (ii) no edge gets sampled at time $t$ i.e., $F_{v_t}= \tF_{v_r} = \bot$. We let $\tF_{v_t}=u$ represent the event that an edge $(u,v_t)$ is both selected i.e., $F_{v_t}=u$ and $F_{v_t}\neq w$ for any $w \in \nut$ and $F_{v_t}=\bot$ represent the event that none of edges were selected. 
    Since \( \mathbb{E}[M_n \mid \mathcal{H}_{t-1}, \tF_{v_t} = u] - \mathbb{E}[M_n \mid \mathcal{H}_{t-1}] \) and \( \bE[M_n \mid \mathcal{H}_{t-1}, \tF_{v_t} = \bot] - {\bE}[M_n \mid {\cH}_{t-1}] \) represent the one-step change in \( M_n \) (the expected matching size restricted to color class \( c \)) when \( \tF_{v_t} = u \) and \( \tF_{v_t} = \bot \), respectively.  
    By the definition of the random variable $\tF_{v_t}$ we have $\Pr[\tF_{v_t}=u]=x_{u,v_t}$ and $\Pr[\tF_{v_t} = \bot] = (1-x(t))$, therefore
    \begin{align}
   & \bE[ |\Delta M_t| \mid \cH_{t-1} ] \notag \\
        &  = \sum_{u\in N(v_t)} x_{u,v_t} \big|  \bE[M_n \mid \mathcal{H}_{t-1}, \tF_{v_t} = u] - {\bE}[M_n \mid {\cH}_{t-1}] \big| + \notag \\ & \quad  (1-x(t))| \bE[M_n \mid \mathcal{H}_{t-1}, \tF_{v_t} = \bot] - {\bE}[M_n \mid {\cH}_{t-1}]|. \label{eq:sample_def} 
    \end{align}
By summing \Cref{eq:sample_def} over all time steps \( t \in [n] \), we obtain the following after applying \Cref{lemma:guts}, and \Cref{lemma:bt}.  
\begin{align}
        \sum_{t\in [n]} \bE[ |\Delta M_t| \mid \cH_{t-1} ]
            &\leq 4M_0 + 2M_0  \label{eq:bt} 
    \end{align}
    The last inequality follows from \Cref{lemma:bt} and \Cref{lemma:guts}. Finally we can conclude that $\sum_{t\in [n]}V_t \leq 12M_0$ by using \Cref{obs:var_ub}.
    \end{proof}
Therefore, by applying the Freedman's inequality (i.e., \cref{thm:freedman}) we get the desired concentration for the size of matching from color class $c$. 
\begin{theorem} \label{thm:concentration_mc}
For any color class \(c\), \cref{alg:ocrs} returns a matching $\scr{M}_c = \scr{M} \cap E_c$ that is sharply concentrated around its expected value \(M_0=\bE[|\cM_c|]\). Specifically, for any \(\delta > 0\), we have that: \begin{equation*} 
    \Pr\Big[ \big|  |\cM_c| -M_0 \big| \geq \delta M_0  \Big]   \leq  2 \exp \Big({- \frac{\delta^2 \sum_{e\in E_c}x_e}{28}} \Big)
\end{equation*}

\end{theorem}
\begin{proof} 
We have \(\Lambda = \max_{t} \Delta M_t \leq 2\). Applying Freedman's inequality from \Cref{thm:freedman} to the martingale \((M_t)_{t \in [n]}\), where \(M_n:=|\cM_c|\) represents the random variable denoting the number of edges of color \(c\) selected by our algorithm in the returned matching, we obtain:

\begin{align*}
    \Pr[ |M_0 -M_n| \geq \lambda] \leq 2\exp \Big(\frac{-\lambda^2}{2(\nu +\lambda (2))} \Big) \leq 2\exp \Big(\frac{-\lambda^2}{2(12M_0 +2\lambda )} \Big).
\end{align*}

{By taking $\lambda = \delta M_0$ we get :
\begin{align*}
     \Pr[ |M_0 -M_n| \geq \delta M_0] \leq  2\exp \Big(\frac{-\delta^2M_0^2}{ 24M_0 +4\delta M_0 } \Big)  \leq 2 \exp \Big(-\frac{\delta^2 M_0}{28} \Big)
\end{align*}
}
as desired. \end{proof}

This concludes that  Property~\ref{eqn:concentration} described in \cref{sec:prelim-prob} can be satisfied. Finally combining Property~\ref{eqn:concentration} together with the fairness constraints in \eqref{eq:coloring_constraint}, we can conclude the following lemma.
\begin{lemma} \label{lemma:delta_prob_fair}
    For any $c\in [\ell]$, and $\delta>0$, \cref{alg:ocrs} returns a matching $\cM$ that satisfies the following: 
    \begin{align*} 
    \Pr \Bigg[ {(1-\delta)} \alpha \leq  \frac{|\cM_c|}{|\cM|} \leq  {(1+\delta)} \beta \Bigg] 
    &\geq 1- 4\exp \Bigg(- \frac{\delta^2\sum_{e\in E_c}x_e}{28}\Bigg)
    \end{align*}
\end{lemma}
\begin{proof}
From \cref{thm:concentration_mc} we know that for any color class $c$ the matching size $|\cM_c|$ is concentrated sharply around the mean. Therefore, we have the following,
\begin{align*}
   \Pr \Bigg[ |\cM_c| \leq \frac{(1-\delta)}{2} \alpha \sum_{e \in E}x_e  \Bigg] \leq   \Pr \Bigg[ |\cM_c| \leq \frac{(1-\delta)}{2} \sum_{e\in E_c}x_e  \Bigg] &\leq \exp \Bigg(- \frac{\delta^2\bE[|\cM_c|]}{28}\Bigg) \\
\end{align*}
A similar bound holds for the random variable $|\cM|$ which is given by 
\begin{align*}
    \Pr \big[  {(1-\delta)} \bE[|\cM|]    \leq |\cM|  \leq {(1+\delta)} \bE[|\cM|] \big] &\geq 1- 2\exp \Bigg(- \frac{\delta^2\bE[|\cM|]}{28}\Bigg) 
\end{align*}

Therefore, we have the following, for any $\delta_1, \delta_2 \geq 0$ such that $\frac{1-\delta_2}{1+\delta_2}\geq 1$,
\begin{align*}\small
      \Pr \left[ \Big(|\cM_c| \geq \frac{(1-\delta_1)}{2} \alpha \sum_{e \in E}x_e \Big) \cap  \Big( |\cM| \leq  \frac{(1+\delta_2)}{2}  \sum_{e\in E}x_e \Big) \right] &\geq 1- \exp \Bigg(\frac{-\delta^2\sum_{e\in E}x_e}{28}\Bigg)- \\ & \quad \exp  \Bigg(\frac{-\delta^2\sum_{e\in E_c}x_e}{28}\Bigg) \\
       \Pr \left[ \frac{|\cM_c|}{|\cM|} \geq \frac{(1-\delta_1)}{(1+\delta_2)} \alpha  \Big)  \right] &\geq 1- 2\exp \Bigg(\frac{-\delta^2\sum_{e\in E_c}x_e}{28}\Bigg)  \\
       \Pr \left[ \frac{|\cM_c|}{|\cM|} \geq (1-\delta) \alpha   \right] &\geq 1- 2\exp \Bigg(\frac{-\delta^2\sum_{e\in E_c}x_e}{28}\Bigg) 
\end{align*} 

Similarly, we know that 
\begin{align*}
   \Pr \left[  |\cM_c| \leq  \frac{(1+\delta)}{2} \beta \sum_{e\in E}x_e \right] \geq \Pr \left[  |\cM_c| \leq  \frac{(1+\delta)}{2} \sum_{e\in E_c}x_e \right] \geq 1- \exp \Bigg(\frac{-\delta^2\sum_{e\in E_c}x_e}{28}\Bigg) 
\end{align*}
Moreover, we know that for any $\delta_1,\delta_2 >0$ such that $\frac{1+\delta_1}{1-\delta_2} \leq 1$ we have: 
\begin{align*}
     \Pr \left[ \Big(  |\cM_c| \leq  \frac{(1+\delta_1)}{2} \beta \sum_{e\in E}x_e  \Big) \cap \Big( {|\cM|} \geq  
     \frac{(1-\delta_2)}{2} \sum_{e \in E}x_e  \Big)    \right]  &\geq 1- \exp \Bigg(\frac{-\delta^2\sum_{e\in E}x_e}{28}\Bigg) - \\ & \quad \exp \Bigg(\frac{-\delta^2\sum_{e\in E_c}x_e}{28}\Bigg)  \\
      \Pr \left[   \frac{|\cM_c|}{|\cM|} \leq  \frac{(1+\delta_1)}{2} \cdot \frac{2}{(1-\delta_2)} \beta   \right] &\geq 1- 2\exp \Bigg(\frac{-\delta^2\sum_{e\in E_c}x_e}{28}\Bigg) \\ 
        \Pr \left[   \frac{|\cM_c|}{|\cM|} \leq  {(1+\delta)} \beta   \right]&\geq  1-2\exp \Bigg(\frac{-\delta^2\sum_{e\in E_c}x_e}{28}\Bigg).
\end{align*}
Therefore, combining the two-sided fairness bounds we get the following, for any $\delta>0$,
\begin{align*}
    \Pr \left[ {(1-\delta)} \alpha \leq  \frac{|\cM_c|}{|\cM|} \leq  {(1+\delta)} \beta \right] \geq 1- f_c(\delta,G)
\end{align*}
where $f_c(\delta,G) =4 \exp \big(  \frac{ -\delta^2\sum_{e\in E_c}x_e}{28} \big).$
\end{proof}
\cref{lemma:delta_prob_fair} proves that \cref{alg:ocrs} returns a $\delta$-$\paf$ matching. This lemma, along with the edge selectability guarantees of the OCRS rounding algorithm as stated in Property~\ref{eqn:rounding_guarantee} completes the proof of \cref{thm:thm_two_sided_fairness}.

\section{Exact fairness for $\beta$-fairness constraints in Bipartite Graphs}
\label{sec:beta-limited}
 In the previous section we showed that our algorithm provides a $1/2$-approximation of the objective, provided we violate 
 the two-sided fairness constraints by multiplicative factors dependent on $\delta$.  Suppose that $\alpha\!=\!0$ and only $\beta$ imposes the fairness constraint, we can argue that by using a simple \emph{LP perturbation} trick, we can remove the $\delta$-violations and make it an \emph{exact} guarantee while still attaining an asymptotic $1/2$-approximation of the objective. More precisely, we satisfy the $\beta$-sided fairness constraint \emph{exactly}. We are also able to provide a better bound on our failure probability. Previously, this was $f_c(\delta,G)$, and it depended on $\sum_{e\in E_c}x_e$.
 We improve this to a new, \emph{smaller} failure probability $f(\epsilon, G)$, which now only depends on $\beta \sum_{e\in E}x_e$ and a tunable parameter $\epsilon$.
 This makes our algorithm more robust to the structure of the input instance. 

If $\alpha = 0$, we modify our algorithm slightly. We begin by solving a perturbed LP where  \(\beta\) in \ref{lp:fairmat} is replaced by \(\beta(1-\epsilon)\) where $0<\epsilon<1$. The perturbed LP requires that any feasible solution must satisfy a stronger $\beta$-fairness constraint, which results in a slightly weaker solution to our problem (see \cref{lemma:objective_perturbed}). Let $\bm{x}$ denote the optimal solution to the perturbed LP i.e., \ref{lp:fairmat} with $\tbeta$. Since our algorithm still executes an OCRS on each $u \in U$ concurrently, \cref{eqn:ocrs_guarantee} ensures that the matching $\scr{M}$ 
has expected weight $\sum_{e \in E} w_e x_e/2$. We now relate $\sum_{e \in E} w_e x_e$
to the optimal solution to \ref{lp:fairmat} when the fairness constraint is $\beta$.
This will allow us to lower bound the expected weight of our matching in terms of $\OPT$ in \Cref{cor:expected_weight}.
(Recall that $\OPT$ is the weight of the optimal matching in the original problem with fairness constraint $\beta.$)

\begin{lemma}\label{lemma:objective_perturbed}
For any $0<\epsilon <1$, the optimal LP solution to \ref{lp:fairmat} with $\tbeta\!=\!(1-\epsilon) \beta$ is at least $ \sum_{e\in E}w_ex_e \geq (1-\epsilon) \sum_{e \in E} w_e y_e$, where $\bm{y} = (y_e)_{e \in E}$
is an optimal LP solution to \ref{lp:fairmat} with $\beta$.
\end{lemma}
\begin{proof} 
Any feasible solution to \ref{lp:fairmat} with \(\alpha=0\) and \(\beta(1-\epsilon) < 1\) is also a feasible solution when \(\alpha=0\) and \(\beta < 1\). 
Let $\bm{x}=(x_e)_{e\in E}$ and $\bm{y}:= (y_e)_{e\in E}$ are optimal fractional solutions to \ref{lp:fairmat} with $\tbeta= \beta(1-\epsilon)$  and $\tbeta=\beta$ respectively. Therefore,
\begin{align}
    \sum_{e \in E}w_ex_e \geq (1-\epsilon)  \sum_{e \in E}w_ey_e.
\end{align}
\end{proof} 
\begin{corollary} \label{cor:expected_weight}
    \cref{alg:ocrs} returns a matching $\cM$ with an expected weight of at least $\frac{1}{2}(1-\epsilon)\OPT.$
\end{corollary}


We next argue that the matching $\cM$ we output satisfies the $\beta$-fairness constraints exactly with a probability of at least \(1 - f(\epsilon, G)\) where $f(\epsilon,G) = 2\exp \Big(- \frac{\epsilon^2\beta \sum_{e\in E}x_e}{28} \Big) $. Note that the failure probability of our algorithm now only depends on the parameter $\epsilon$ and $\sum_{e\in E}x_e$ (see \cref{lemma:beta_guarantee}). It is also reduced by a factor of $2$, since we can use a one-sided version of Freedman's inequality when $\alpha =0$ and $\beta > 0$. Here, \(\sum_{e\in E}x_e\) refers to the size of optimal fractional matching of \ref{lp:fairmat} with \(\tbeta = \beta(1-\epsilon)\). 




\begin{lemma}  \label{lemma:beta_guarantee}
    Given an optimal solution $\bm{x}$ to \ref{lp:fairmat} with $\tbeta= (1-\epsilon)\beta$, then we can recover a solution that satisfies the constraints exactly i.e., no violations in $\beta$-fairness constraints. Formally, we can say that for any color $c$, and $0<\epsilon<1$, the size of the matching restricted to color class $c$ is at most $\beta |\cM|$ with high probability,
    \begin{align*}
   \Pr \left[ \frac{|\cM_c|}{|\cM|} \leq \beta \right]  \geq 1- 2\exp\Big( \frac{-{\epsilon^2}\beta\sum_{e\in E}x_e}{28} \Big).
\end{align*}
\end{lemma}

\begin{proof}
    Recall that $\bm{x}$ denotes the optimal solution to \ref{lp:fairmat} with $\Tilde{\beta} = (1-\epsilon)\beta$ and $\cM$ is the matching returned after the rounding $\bm{x}$ by \cref{alg:ocrs}. Therefore we know that 
    \[\bE[|\cM_c|] \leq \sum_{e\in E_c}x_e \leq \Tilde{\beta}\sum_{e\in E}x_e = (1-\epsilon)\beta \sum_{e\in E}x_e.\]  
Using \cref{lemma:muh}, we can say that  
    \begin{align*}
        \Pr \left[ |\cM_c| \geq (1+\delta_1) (1-\epsilon)\beta \sum_{e \in E}x_e \right] \leq \exp \Bigg( -\frac{\tbeta \delta_1^2 \sum_{e\in E}x_e }{  28} \Bigg).
    \end{align*}
Additionally, we have that $|\cM| \leq  (1-\delta_2)\sum_{e\in E}x_e$ with probability $\exp \big( \frac{-\delta_2^2 \sum_{e\in E}x_e}{2} \big)$ which follows from Chernoff-like bounds of \citep{panconesi1997randomized} and the facts $\bE[Z_{u,v}] \leq \Pr[F_{v}=u]=x_{u,v}$ for any $e = (u,v)\in E$.
Therefore, we have the desired exact fairness as follows,
\begin{equation*}
        \Pr \left[ \frac{|\cM_c|}{|\cM|} \geq \frac{(1-\epsilon)(1+\delta_1)}{(1-\delta_2)}{\beta}  \right] \leq 2\exp\Bigg( -\frac{\delta_1^2 (1-\epsilon)}{28}\beta \sum_{e\in E}x_e  \Bigg) \notag
\end{equation*}
We know that for any $\epsilon \geq \frac{\delta_1+\delta_2}{1+\delta_1},$ the term $\frac{(1-\epsilon)(1+\delta_1)}{(1-\delta_2)}\leq  1$  implying that
\begin{equation*}
   \Pr \left[ \frac{|\cM_c|}{|\cM|} \leq \beta \right]  \geq 1- 2\exp\Bigg( -\frac{\delta_1^2}{28(1-\epsilon)}\beta \sum_{e\in E}x_e  \Bigg) \geq 1- 2\exp\Bigg( -{\epsilon^2}\frac{\beta}{28}\sum_{e\in E}x_e  \Bigg) 
\end{equation*}
The last inequality assumes $\delta_1\geq (1-\epsilon)\epsilon.$
\end{proof}

Therefore we can conclude that the probability with which our algorithm fails to satisfy $\beta$-fairness constraints is now tightened from $f_c(\delta,G)$ which depends on $\sum_{e\in E_c}x_e$ of individual color classes, to a global quantity $\beta \sum_{e\in E}x_e$ in $f(\epsilon,G).$ 
This concludes the proof of \cref{thm:one_sided_fairness}. 

\subsection{Brute Force Algorithm for Small Instances}\label{sec:brute_force_alg}
Given an instance of \pfm{}, recall that the performance of \Cref{alg:ocrs}
improves the larger $\sum_{e \in E} \beta x_e$ is, where $\bm{x} = (x_e)_{e \in E}$
is an optimal solution to \ref{lp:fairmat} with $\tilde{\beta} = (1-\eps) \beta$.

We now handle the case where $\sum_{e\in E} \beta x_e \le C$ by running a simple brute-force algorithm. In order to see how to do this,
let $\bm{y}=(y_e)_{e \in E}$ be an optimal solution to \ref{lp:fairmat} with $\beta$. (We don't actually use $\bm{y}$ in our algorithm, it is just for the analysis). Define $L = \min_{e \in E} w_e $ and $U = \max_{e \in E} w_e$. Then,
$$
U \sum_{e \in E} x_e \ge \sum_{e \in E} w_e x_e \ge (1 - \eps) \sum_{e \in E} w_e y_e \ge (1 - \eps) L \sum_{e \in E} y_e,
$$ where the second inequality from the left uses \Cref{lemma:objective_perturbed}, and the rest use the trivial upper and lower bounds of $U$ and $L$. We thus have that
\begin{equation} \label{eqn:lp_comp}
\frac{U}{L(1- \eps)} \sum_{e \in E} x_e \ge  \sum_{e \in E} y_e.
\end{equation}
Thus, if $\sum_{e \in E} \beta x_e \le C$, then 
$\sum_{e \in E} y_e \le \frac{U}{L(1-\eps)} \frac{C}{\beta}$ via \eqref{eqn:lp_comp}.
Let us now suppose that $M_{OPT}$ is a maximum weighted matching with respect
to fairness constraint $\beta$. Then,
$L |M_{\OPT}| \le \sum_{e \in M_{\OPT}} w_e \le \sum_{e \in E} w_e y_e \le U\sum_{e \in E} y_e \le U \frac{U}{\beta L} \frac{C}{1-\eps}$,
where the last inequality uses the assumed inequality $\sum_{e \in E} y_e \le \frac{U}{\beta L} \frac{C}{1-\eps}$. It follows that
\begin{equation} \label{eqn:matching_upper_bound}
|M_{OPT}| \le \frac{U^2}{L^2} \frac{C}{\beta (1 -\eps)}.
\end{equation}
Thus, if we run a brute force algorithm where we check matchings which are fair with respect
to $\beta$, and which contain at most $\frac{U^2}{L^2} \frac{C}{\beta (1 -\eps)}$ edges,
then \eqref{eqn:matching_upper_bound} ensures that we'll return an optimal matching.

\section{Conclusion and Future work}
In contrast to our current algorithm which processes vertices in a fixed or arbitrary order, Random Order Contention Resolution Scheme (RCRS) from \cite{adamczyk2018random} provides better selection guarantees where vertices are processed in a random order. More precisely, RCRS achieves an approximation ratio of \(1 - \frac{1}{e}\). However, it appears to be much more challenging to compute the one-step variance of the martingale at each time \(t\), making it difficult to derive a \emph{useful bound} for the sum of variances \(\sum_{t \in [n]} \var(|\Delta M_t| \mid \cH_{t-1})\). 

Thus, it remains an interesting open problem to determine if we can overcome this obstacle to improve the approximation factor from \(1/2\) to \(1 - \frac{1}{e}\). Moreover, adapting the analysis to general graphs is also an interesting future direction.

\section*{Acknowledgments}
Sharmila Duppala, Nathaniel Grammel, Juan Luque, and Aravind Srinivasan were supported in part by NSF award number CCF-1918749.

\bibliographystyle{abbrvnat}
\bibliography{references}
\appendix

\section{Example demonstrating that $|\Delta M_t| \geq 2 - \epsilon$ in the Worst Case}
\begin{observation}
For any $t \geq 1$ and $\epsilon > 0$, the one-step change in the martingale can be as large as $2 - \epsilon$, i.e., $\Delta M_t \geq 2 - \epsilon$.
\end{observation}
\begin{proof}
    We establish this by constructing an example where the one-step change in the martingale is precisely $2 - \epsilon$. Consider the \cref{eq:worst_case_delta_mt} where, at time $t = 1$, vertex $v_t$ proposes to vertex $u$. In this case, the one-step change in the matching from the blue color class, $\cM_{\text{blue}}$, given $\cH_{t-1}$, is $2 - \epsilon$.
 \end{proof}
\begin{example} \label{eq:worst_case_delta_mt}
    Consider an instance of proportional matching in an edge-colored graph $G = (U, V, E)$, where $E = E_{\text{red}} \cup E_{\text{blue}}$. Suppose the optimal fractional solution has $x_{u, v_k} = x_{w, v_t} = 1 - \epsilon$ and $x_{u, v_t} = \epsilon$. 
\end{example}
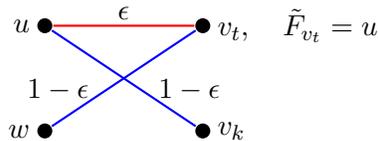
\begin{figure}[H]
  \centering 
  \begin{tikzpicture}[scale=0.7,thick,inner sep=2pt]
    \node[fill,circle,label=right:{ $v_t, \quad \tF_{v_t} = u$}] (vt) at (0, 0) {};
    \node[fill,circle,label=right:{ $v_k$}] (vk) at (0, -2) {};
    \node[fill,circle,label=left:{ $u$}] (u) at (-3, 0) {};
    \node[fill,circle,label=left:{ $w$}] (uu2) at (-3, -2) {};
    \draw[solid,color=red] (u) -- (vt) node[midway,above] {\color{black}$\epsilon$};
    \draw[solid,color=blue] (u) -- (vk) node[pos=0.8,above,xshift=3mm] {\color{black}$1-\epsilon$};
    \draw[solid,color=blue] (uu2) -- (vt) node[pos=0.2,above,xshift=-3mm] {\color{black} $1-\epsilon$};
  \end{tikzpicture}
  \caption{At $t=1$ when $v_t$ proposes to $u$ instead of $w$ and $\tF_{v_t}=u$, then the expected change in the matching restricted to blue edges, $\cM_{\text{blue}}$, is bounded by $ 2(1 - \epsilon)$. }
  \label{fig:example_mt_2}
\end{figure}

\section{Omitted proofs from \cref{sec:alg-analysis}} \label{appendix:martingale}

\begin{lemma}\label{lemma:wose_case_mt}
For any $t\geq 1$ when vertex $v_t$ is processed, the worse case one-step change given $\tF_{v_t} = u$ is as follows:
\[ \max_{u \in N(v_t)} |\bE[M_n \mid \cH_{t-1}, \tF_{v_t}= u] - \bE[M_n \mid \cH_{t-1} ]| \leq  2 \]
\end{lemma}
\begin{proof}
Recall that $Z_{w,v_t}$ denotes the indicator random variable that the edge $(w,v_t)$ is selected in the matching $\cM.$ 
Given the definition of \(M_t\) in \cref{eq:def_mt} and conditioned on \(\tF_{v_t} = u\) for any \(u \in N(v_t)\), the one-step change is given by:
\begin{align}
    &\bE\Big[  \sum_{k \geq 1} Z (v_k)  \mid \cH_{t-1}, \tF_{v_t}=u \Big]- \bE\Big[  \sum_{k \geq 1}Z (v_k)  \mid \cH_{t-1} \Big]  \notag \\
    &= \sum_{k\geq t} \bE\big[   Z (v_k)  \mid \cH_{t-1}, \tF_{v_t}=u \big]- \bE\big[  Z (v_k)  \mid \cH_{t-1} \big] \notag \\
    &= \sum_{k\geq t} \sum_{w \in N(v_k)} \awvk  \Big( \bE\big[   Z_{w,v_k}  \mid \cH_{t-1}, \tF_{v_t}=u \big]- \bE\big[  Z_{w,v_k}  \mid \cH_{t-1} \big]  \Big) \label{eq:desire}
\end{align}

Notice that for any $(w,v_k)$ with $w \notin N(v_t)$ are not affected by the conditional on $\tF_{v_t}=u$ i.e., 
\[ \bE\big[   Z_{w,v_k}  \mid \cH_{t-1}, \tF_{v_t}=u \big]= \bE\big[  Z_{w,v_k}  \mid \cH_{t-1} \big].\]
Therefore, we need to bound the term \(\bE[ Z_{w,v_k} \mid \cH_{t-1}, \tF_{v_t}=u] - \bE[ Z_{w,v_k} \mid \cH_{t-1}]\) for any \(k \geq t\) and \(w \in N(v_t)\) in order to obtain the desired bound stated in the lemma. Additionally, these edges can be classified into one of the following cases. 
\paragraph{Case 1:} \( k = t \), \( w = u \), and \(\tF_{v_t} = u\) 
\begin{equation}
    \bE\big[ Z_{u,v_t} \mid \cH_{t-1}, \tF_{v_t}=u \big] - \bE\big[ Z_{u,v_t} \mid \cH_{t-1} \big] 
    \leq \Pr\big[ u \text{ is safe at } t \mid \cH_{t-1} \big](1 - \Pr[\tF_{v_t}=u \mid \cH_{t-1}] ) \leq 1 -\frac{ x_{u,v_t}}{2}. \label{eq:uvtu}
\end{equation}
The last inequality is from \cref{eqn:ocrs_guarantee}. 
\paragraph{Case 2:} \( k = t \), \( w \neq u \), and \(\tF_{v_t} = u\)
\begin{align}
    \bE\big[ Z_{w,v_t} \mid \cH_{t-1}, \tF_{v_t}=u \big] - \bE\big[ Z_{w,v_t} \mid \cH_{t-1} \big] &\geq 0 - x_{w,v_t} \Pr\big[ w \text{ is safe at } t \mid \cH_{t-1} \big] \geq -x_{w,v_t} \label{eq:wvtu}
\end{align}

\paragraph{Case 3:} \( k > t \), \( w = u \), and \(\tF_{v_t} = u\)
\begin{align}
    \bE\big[ Z_{u,v_k} \mid \cH_{t-1}, \tF_{v_t}=u \big] - \bE\big[ Z_{u,v_k} \mid \cH_{t-1} \big] &= 0 - x_{u,v_k} \Pr\big[ u \text{ is safe at } t \mid \cH_{t-1} \big] \geq - x_{u,v_k} \label{eq:wvku}
\end{align}

\paragraph{Case 4:} \( k > t \), \( w \neq u \), and \(\tF_{v_t} = u\)
\begin{align}
    &\bE\big[ Z_{w,v_k} \mid \cH_{t-1}, \tF_{v_t}=u \big] - \bE\big[ Z_{w,v_k} \mid \cH_{t-1} \big] \\
    &\leq x_{w,v_k} \left( \Pr\big[ w \text{ is safe at } k \mid \cH_{t-1}, \tF_{v_t}=u \big] - \Pr\big[ w \text{ is safe at } k \mid \cH_{t-1} \big] \right) \notag \\
    &\leq x_{w,v_k} x_{w,v_t} \label{eq:safety_k_u} 
\end{align}
Inequality~\eqref{eq:safety_k_u} follows from \cref{lemma:safety_k_u}. Hence, substituting the bounds from the four cases into \cref{eq:desire} we have:
\begin{align*}
    &\bE\left[ \sum_{k \geq 1} Z(v_k) \mid \cH_{t-1}, \tF_{v_t}=u \right] - \bE\left[ \sum_{k \geq 1} Z(v_k) \mid \cH_{t-1} \right] \\
    &\leq \auvt \big(1 - \frac{x_{u,v_t}}{2}\big) - \sum_{w \in N(v_t) \setminus \{u\}} \awvt x_{w,v_t} -   \\ & \quad \sum_{k > t} \auvk x_{u,v_k} + \sum_{k > t} \sum_{w \in N(v_t) \setminus \{u\}} \awvk x_{w,v_k} x_{w,v_t} \\
    &\leq \auvt \big(1 - \frac{x_{u,v_t}}{2}\big) -  \sum_{w \in N_{\tilde{u}}(v_t)} \awvt x_{w,v_t} - \sum_{k > t} \auvk x_{u,v_k} + \\
    &\quad \sum_{k > t} \sum_{w \in N_{\tilde{u}}(v_t) \setminus \{u\}} \awvk x_{w,v_k} x_{w,v_t} \\
    &\leq 1 -\frac{x_{u,v_t}}{2} - \sum_{w \in N_{\tilde{u}}(v_t)} \awvt x_{w,v_t} - \sum_{k > t} \auvk x_{u,v_k} + (1 - x_{u,v_t}) \\
    &\leq 2.
\end{align*}
The second-to-last inequality follows from the fact:
\begin{align*}
    \sum_{k > t} \sum_{w \in N(v_t) \setminus \{u\}} \awvk x_{w,v_k} x_{w,v_t} = \sum_{w \in N(v_t) \setminus \{u\}} x_{w,v_t} \sum_{k > t} \awvk x_{w,v_k} \leq \sum_{w \in N(v_t) \setminus \{u\}} x_{w,v_t} 
    \leq 1 - x_{u,v_t}.
\end{align*}
The final inequality is due to:
\begin{align*}
    \sum_{w \in N(v_t)} \awvt x_{w,v_t} + \frac{1}{2}x_{u,v_t} + \sum_{k > t} \auvk x_{u,v_k} \leq 2
\end{align*}
and this concludes the proof.
\end{proof}
\begin{lemma}\label{lemma:wose_case_mt_2}
For any $t\geq 1$ when vertex $v_t$ is processed, the worse case one-step change given $\tF_{v_t} = \bot$ is as follows:
\[ \bE[M_n \mid \cH_{t-1}, \tF_{v_t}= \bot] - \bE[M_n \mid \cH_{t-1} ] \leq  2 \]
\end{lemma}

\begin{proof}
    Following the definitions of $M_t$ in \cref{eq:def_mt} we have:
\begin{align} 
 \bE[M_n \mid \cH_{t-1},\tF_{v_t} = \bot] - \bE[ M_n \mid \cH_{t-1}]  
    &=  \sum_{k \geq 1} \bE[ Z(v_k) \mid \cH_{t-1},\tF_{v_t} = \bot] - \bE[ Z(v_k) \mid \cH_{t-1}] \notag \\   
    &=  \sum_{k \geq t} \sum_{ u \in N(v_t)} \auvk \big( \bE[Z_{u,v_k} \mid \cH_{t-1},\tF_{v_t} = \bot ] - \bE[  Z_{u,v_k} \mid \cH_{t-1} ] \big) \label{eq:important}  
\end{align} 
Notice that for any $(u,v_k)$ with $u \notin N(v_t)$ are not affected by the conditional on $\tF_{v_t}=\bot$ i.e., 
$$ \bE\big[   Z_{u,v_k}  \mid \cH_{t-1}, \tF_{v_t}=\bot \big]= \bE\big[  Z_{u,v_k}  \mid \cH_{t-1} \big].$$
Therefore, we need to bound the term \(\bE[ Z_{u,v_k} \mid \cH_{t-1}, \tF_{v_t}=\bot] - \bE[ Z_{u,v_k} \mid \cH_{t-1}]\) for any \(k \geq t\) and \(u \in N(v_t)\) in order to obtain the desired bound stated in the lemma. Additionally, these edges can be classified into one of two  cases. 
\paragraph{Case 1:}  \( k = t \), \( u \in N(v_t)\) and \(\tF_{v_t} = \bot \)
\begin{align*}
    \bE\big[ Z_{u,v_t} \mid \cH_{t-1}, \tF_{v_t}=\bot \big] - \bE\big[ Z_{u,v_t} \mid \cH_{t-1} \big] &= 0 -  
    \Pr\big[ u \text{ is safe at } t \mid \cH_{t-1} \big] \Pr[\tF_{v_t}=u ]\\
    &\geq  - x_{u,v_t}
\end{align*}

\paragraph{Case 2:} \( k> t \), \( u \in N(v_t) \), and \(\tF_{v_t} = \bot \)
\begin{align}
    &\bE\big[ Z_{u,v_k} \mid \cH_{t-1}, \tF_{v_t}=\bot \big] - \bE\big[ Z_{u,v_k} \mid \cH_{t-1} \big] \notag \\
    &=  \Pr[ \tF_{v_k} = u ] \big( \Pr[ u \text{ is safe at }k \mid \cH_{t-1}, \tF_{v_t}=\bot ] - \Pr[ u \text{ is safe at }k \mid \cH_{t-1}]  \big) \notag \\
    &\leq x_{u,v_k} x_{u,v_t}. \label{eq:safety_k_u_2}
\end{align}
Inequality~\eqref{eq:safety_k_u_2} follows from \cref{lemma:safety_k_u_2}. Hence, substituting the bounds from the four cases into \cref{eq:important} we have:
\begin{align}
    \bE[M_n \mid \cH_{t-1},\tF_{v_t} = \bot] - \bE[ M_n \mid \cH_{t-1}]   
    &\leq \sum_{k \geq t} \sum_{ u \in N(v_t) 
\setminus \{u\}} \auvk  x_{u,v_k} x_{u,v_t}  + \sum_{u \in N(v_t)}\auvt x_{u,v_t} \notag \\
    &\leq \sum_{ u \in N(v_t) \setminus \{u\}} \sum_{k \geq t}  \auvk  x_{u,v_k} x_{u,v_t} +  \sum_{u \in N(v_t)}\auvt x_{u,v_t}  \label{eq:botmn}\\
    &\leq 1 +1 =2. \notag
\end{align}
The last inequality is due to \[\sum_{ u \in N(v_t) \setminus \{u\}} \sum_{k \geq t}  \auvk  x_{u,v_k} x_{u,v_t} =\sum_{ u \in N(v_t) \setminus \{u\}} x_{u,v_t} \sum_{k \geq t}  \auvk  x_{u,v_k}  \leq \sum_{ u \in N(v_t) \setminus \{u\}} x_{u,v_t} \leq 1.  \]
This completes the proof. \end{proof}

\begin{lemma} \label{lemma:gut}
Suppose that $\tF_{v_t} =u$ represent the event that $v_t$ proposed to $u$ and $A_{u,v_t}=1$ in step $t$, then
\begin{align*}
    & \sum_{u \in N(v_t)} x_{u,v_t} \left| \bE[M_n \mid \cH_{t-1},\tF_{v_t}=u] - \bE[M_n \mid \cH_{t-1}]  \right|  \\
    & \leq \sum_{\mathclap{u \in N(v_t)}} x_{u,v_t} \Big( \auvt +    \sum \limits_{\mathclap{w \in N(v_t)}} \big( a_{w,v_t} x_{w,v_t} + \sum_{k >t} (\auvk x_{u,v_k} + a_{w,v_k} x_{w,v_k}x_{w,v_t} ) \big) \Big)
\end{align*}
\end{lemma}
\begin{proof}
By applying \Cref{lemma:one_step_one_vertex_diff} we have the following, 
\begin{align}
    &\sum_{u \in N(v_t)}x_{u,v_t} \Big| \bE[M_n \mid \cH_{t-1},\tF_{v_t}=u] - \bE[M_n \mid \cH_{t-1}] \Big| \notag \\
     &=\sum_{u \in N(v_t)}x_{u,v_t} \Big| \bE \Big[ \sum_{k\geq t} Z(v_k)  \mid \cH_{t-1},\tF_{v_t}=u \Big] - \bE \Big[\sum_{k\geq t} Z(v_k) \mid \cH_{t-1}\Big] \Big| \notag  \\
    &\leq\sum_{u \in N(v_t)}x_{u,v_t} \sum_{k\geq t} \Big| \bE \big[  Z(v_k)  \mid \cH_{t-1},\tF_{v_t}=u \big] - \bE \big[Z(v_k) \mid \cH_{t-1}\big] \Big| \notag \\
   & \leq 
     \sum_{u \in N(v_t)} x_{u,v_t} \big( \auvt +  \sum \limits_{w \in N(v_t)} \awvt x_{w,v_t} \big) +  \notag \\ 
     & \quad \sum_{u \in N(v_t)} x_{u,v_t} \sum_{k>t}\big( \auvk x_{u,v_k}+  \sum \limits_{w \in \nut} \awvk x_{w,v_k} x_{w,v_t}   \big) \label{eq:lemma_use}
    \\
&= \sum_{u \in N(v_t)} x_{u,v_t} \big( \auvt +    \sum \limits_{w \in N(v_t)} \awvt x_{w,v_t} + \sum_{k>t}( \auvk x_{u,v_k} +  \sum \limits_{w \in N(v_t)}  \awvk x_{w,v_k}x_{w,v_t}) \big) \label{eq:drop_u}
\end{align}
\Cref{eq:lemma_use} follows directly from \Cref{lemma:one_step_one_vertex_diff}.  Thus, we obtain the desired bound.
\end{proof}

\section{Omitted proofs from \cref{sec:beta-limited}}
In order to attain the desired one-sided concentration inequality as stated in \Cref{lemma:beta_guarantee}, we will a variant of Freedman's inequality for supermartingales. This inequality is the same as \Cref{thm:freedman}, except that since it only controls the upper tail of the random variable, the probability term is reduced by a factor of $2$.
\begin{theorem}[\citep{bohman2019dynamicconcentrationtrianglefreeprocess}]
Suppose $(M_t)_{t\in [n]}$ is a supermartingale such that $|\Delta M_t| \leq \Lambda$ for any $1 \leq t \leq n$, and $ \sum_{t\in [n]} \bE[ \Delta M_t^2
\mid \cH_{t-1}] \leq \nu.$ Then, for all $\lambda > 0$, 
\begin{equation*}
    \Pr \big[ M_n  \ge M_0  + \lambda \big] \leq \exp \Bigg( \frac{-\lambda^2}{2(\nu+ \lambda \Lambda) } \Bigg)
\end{equation*}
\end{theorem}
\begin{lemma} \label{lemma:muh}
Suppose that $(M_t)_{t\geq 0}$ is the Doob martingale defined in \cref{sec:alg-analysis} and $M_0 \leq \mu_{H} $ i.e., we know an upper bound on the expected size of the matching from color class $c$, then
    \[\Pr[ M_n \geq (1+\delta)\mu_{H}] \leq \exp \Bigg( \frac{-\delta^2\mu_H}{28} \Bigg) \]
\end{lemma}
\begin{proof}
        Consider the following sequence of random variables $(S_t)_{t\geq 0}$ where $S_t:= M_t - \frac{\tbeta}{2} \sum_{e\in E}{x_e}$. Clearly, $(S_t)_{t\geq 0}$ is a supermartingale since $(M_t)_{t\ge 0}$ is a martingale and $M_0  \leq \frac{ \tbeta}{2}\sum_{e \in E} x_e$. The latter is true because, 
\[ M_0 = \bE[|\cM_c|]=  \sum_{e\in E_c}x_e/2  \leq \tbeta \sum_{e\in E}x_e/2 = \mu_H \]
This implies the following upper tail bound on $S_n$:
\begin{align*}
  \Pr[S_n  \geq  \lambda ]  \leq 
 \Pr[S_n  \geq S_0+ \lambda ] &\leq \exp \Bigg( -\frac{\lambda^2}{ 2 \big(\sum_{t \in [n]}\bE(|\Delta S_t|^2\mid \cH_{t-1}) + \lambda  \big)} \Bigg) \\
\Pr[ M_n \geq \lambda + \tbeta \sum_{e\in E}x_e/2  ] &\leq \exp \Bigg( -\frac{\lambda^2}{2 \big( \sum_{t \in [n]}\bE(2|\Delta M_t|\mid \cH_{t-1}) + \lambda\big) } \Bigg).
\end{align*}
The last inequality is due to the fact that $|\Delta S_t| = |\Delta M_t|.$
Therefore, we have that
\begin{align*}
     \Pr[M_n \geq \tbeta \sum_{e\in E}x_e/2  + \lambda] \leq \Bigg( -\frac{\lambda^2}{2 \big(12M_0 + 2\lambda\big) } \Bigg).
\end{align*}
By substituting $\lambda  = 
\delta \frac{\tbeta}{2}\sum_{e\in E}x_e  \geq \delta M_0 $,
 \begin{align*}
     \Pr[M_n \geq  (1+\delta) \mu_H ] &\leq \exp \Bigg( -\frac{\lambda^2 \delta }{2 \big(10\lambda + 2\lambda \delta \big) } \Bigg) \\
     & = \exp \Big( -\frac{\lambda \delta  }{28   } \Big)  \\
      &=\exp \Big( -\frac{ \delta^2 \mu_H }{  28} \Big)
 \end{align*}
 This concludes the proof.
\end{proof}

\section{Construction of an instance of proportional fair matching where Azuma-Hoeffding inequality fails to achieve concentration on $|\cM_c|$} 

\begin{theorem}[Azuma–Hoeffding inequality]
   Suppose that $(M_t)_{t\geq 0}$ is a martingale and $|M_t -M_{t-1}| \leq c_t$ for any $1\leq t \leq n$. Then for any $\delta>0$ 
   \begin{equation*}
        \Pr[ |M_{n}-M_{0}|\geq \delta ]\leq 2 \exp \left(\frac{-\delta ^{2}}{2\sum _{t=1}^{n}c_{t}^{2}}\right).
   \end{equation*}
\end{theorem}
\label{ex:azuma_fails}
\begin{example} 
Consider an edge-colored star graph \( G \) with 2 colors and \( n+2 \) vertices. Let \( u \) be the central vertex, and let the remaining vertices be \( v_1, v_2, \dots, v_{n+1} \). The edges partitions are as follows: \( E_{\text{blue}} = \{(u, v_t) \mid t = 1, \dots, n\} \) and \( E_{\text{red}} = \{(u, v_{n+1})\} \). An optimal solution to the linear program (LP) \ref{lp:fairmat} is given by \( x_{u, v_t} = \frac{\epsilon}{n} \) for \( 1 \leq t \leq n \) and \( x_{u, v_{n+1}} = 1 - \epsilon \).\end{example}

Let us fix the color class \emph{red} and $\cM_{\text{red}}$ denote $\cM \cap E_{\text{red}}$ where $\cM$ is the matching returned by our algorithm. Let the martingale $(M_t)_{t\geq 0}$ correspond to the matching restricted to color class $c.$ Then we have the following observation.
\begin{observation}
There exists an instance of proportional fair matching where the worst-case sum \(\sum_{t \geq 1} c_t\) can grow as large as \(\Theta(n)\), where \(n\) represents the number of vertices and \(c_t\) are positive constants such that \(|M_t - M_{t-1}| \leq c_t\).
\end{observation}
\begin{proof}
For each $1\leq t \leq n$, when $v_t$ is processed we know that the one-step change in the $\cM_c$ is given by $|\Delta M_t| = 1-\epsilon.$ Therefore we have $\sum_{t\in [n]}c_t^2 \leq \sum_{t\in [n]}c_t = n(1-\epsilon)$.   
\end{proof}
Therefore, we can conclude that Azuma's inequality may fail to provide concentration for \( |\cM_c| \) since \( \mathbb{E}[|\cM_c|] \leq n_c \), where \( n_c \) represents the total number of vertices with at least one incident edge from color class \( c \).

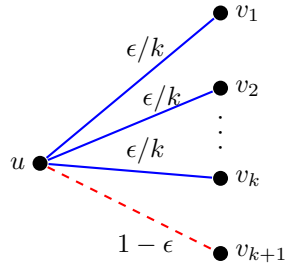
\begin{figure}[H]
  \centering \small
  \begin{tikzpicture}[scale=0.8,thick,inner sep=2pt]
    \node[fill,circle,label=right:{$v_1$}] (v1) at (0, 2.5) {};
    \node[fill,circle,label=right:{$v_{2}$}] (v2) at (0, 1.25) {};
    \node () at (0, 0.75) {$\cdot$};
    \node () at (0, 0.5) {$\cdot$};
    \node () at (0, 0.25) {$\cdot$};
    \node[fill,circle,label=right:{ $v_{k}$}] (vt) at (0, -.25) {};
    \node[fill,circle,label=right:{ $v_{k+1}$}] (vk) at (0, -1.5) {};
    \node[fill,circle,label=left:{ $ u$}] (u) at (-3,0) {};
    \draw[solid,color=blue] (u) -- (v1) node[pos=0.5,above,yshift=0.25cm,xshift=0.2cm,black] {$\epsilon/k$};
    \draw[solid,color=blue] (u) -- (v2) node[pos=0.6,above,yshift=0.01cm,xshift=0.2cm,black] {$\epsilon/k$};
    \draw[solid,color=blue] (u) -- (vt) node[pos=0.5,above,yshift=0.01cm,xshift=0.2cm,black] {$\epsilon/k$};
    \draw[dashed,color=red] (u) -- (vk) node[pos=0.5,below,yshift=-0.27cm,xshift=0.2cm,black] {$1-\epsilon$};

  \end{tikzpicture}
  \caption{In the edge-colored graph \( G \), we have \(|\Delta M_t| = 1-\epsilon \) at each time \( t \). Therefore, the sum \(\sum_{t \in [n]} |\Delta M_t|\) is \( \Theta(n) \).}\label{fig:worst_case_c_t}
\end{figure}

\end{document}